\documentclass[11pt]{article}
\usepackage[utf8]{inputenc}
\usepackage[T1]{fontenc}
\usepackage[margin=1in]{geometry}
\usepackage[round,authoryear]{natbib}
\usepackage{amsmath,amssymb,amsthm,mathtools}
\usepackage{microtype}
\usepackage{enumitem}
\usepackage{booktabs,tabularx}
\usepackage{xcolor}
\usepackage{tikz}
\usetikzlibrary{arrows.meta,positioning}
\usepackage{xurl}
\usepackage{hyperref}
\hypersetup{colorlinks=true,linkcolor=blue!45!black,citecolor=blue!45!black,urlcolor=blue!45!black,
pdftitle={A Relative-Computability Theory of Self-Improving Agents},
pdfauthor={Chien-Ping Lu},
pdfsubject={LLM-based agents, effective self-revision, answering contracts, and certified acquisition},
pdfkeywords={recursive self-improvement, LLM-based agents, relative computability, oracle semantics, Turing degree, query learning, certified acquisition}}
\newtheorem{theorem}{Theorem}[section]
\newtheorem{proposition}[theorem]{Proposition}
\newtheorem{corollary}[theorem]{Corollary}

\theoremstyle{definition}
\newtheorem{definition}[theorem]{Definition}
\newtheorem{example}[theorem]{Example}
\theoremstyle{remark}

\newcommand{\T}{\leq_T}
\newcommand{\strictT}{<_T}
\newcommand{\C}{\mathcal{C}}
\newcommand{\N}{\mathbb{N}}
\newcommand{\bits}{\{0,1\}}
\newcommand{\pair}[2]{\langle #1,#2\rangle}
\setlist{nosep,leftmargin=1.5em}

\title{A Relative-Computability Theory\\of Self-Improving Agents}
\author{Chien-Ping Lu}
\date{}

\begin{document}
\maketitle

\begin{abstract}
\noindent Self-improving agents can acquire a procedure without knowing when its supporting teacher can safely be removed. We extend an oracle-relative model of evolving agents with explicit acquisition and completion contracts, separating learning a solver from certifying its independence. For a known finite-dimensional Boolean feature language, a query history determines exactly the answers in its feature span. A guarded protocol permits query-proposer revision while guaranteeing completion with the optimal number of teacher calls, and a supplied separating specification gives a finite-family extension. The resulting solver answers every input without its teacher and remains in the baseline computability layer. The central obstruction concerns supplying the completion specification: no baseline-effective procedure can extract a valid span or separating specification uniformly from arbitrary total programs, already in one feature dimension or with two candidate descriptions. Classical relative computation and limiting-revision results locate these claims without imposing a ceiling on practical improvement. Controlled experiments distinguish exact teacher removal from low-risk stopping and measure adaptation benefits and costs against a stronger greedy baseline. The framework makes the information required to acquire, certify, and deploy a capability explicit.
\end{abstract}

\section{Introduction}
An agent can use a teacher to acquire a rule, write a reusable solver, and revise the procedure that chooses its next queries. Successful tests then raise a specific question: when can the teacher be removed with a guarantee that the acquired solver is correct on every future input? Learning a useful procedure, identifying the right procedure, and recognizing that identification is complete require different evidence. A theory of self-improvement must account for all three.

We study this question within relative computability. A declared resource $A$ specifies the answers available to the complete evolving computation, including revisions to its solver, evaluator, and improver. The resource-relative class $\mathcal C(A)$ contains all decision tasks computable from $A$; a particular learner starts with only part of the information needed to select a solver from that class. A task can therefore be computable relative to $A$ while its unknown target still requires teacher interaction. Certified acquisition addresses this selection and completion problem.

Our positive result gives a terminating protocol for a known Boolean feature language. Teacher replies determine a subspace of the supplied feature span, and equality of those spans certifies all future answers. A protected guard accepts informative queries from a mutable proposer and falls back to a supplied basis when proposals fail. The protocol attains the classical optimal rank query bound while allowing effective proposer revision. A finite-family extension replaces the feature span with a separating specification for supplied hypotheses. These results make explicit the structural information on which teacher removal depends.

Our complementary obstruction concerns obtaining that information. No baseline-effective method can supply the required span or separating specification uniformly from arbitrary total program descriptions. The reductions already work with a one-dimensional feature map or two candidate descriptions. Using a correct finite specification is effective; uniformly producing one can require stronger information. This distinction is the central connection between the constructive acquisition protocol and the computability analysis.

Membership-query learning supplies the classical acquisition setting \citep{Angluin1988Queries}. Universal simulation, resource joins, and relative halting and limit results supply the computational foundation \citep{Turing1936,Turing1939,Rogers1967,Soare2016}. We integrate these ingredients into contracts for evolving agents; the rank bound and elimination argument are not new learning algorithms. Contemporary systems that modify agent programs and improvement mechanisms motivate accounting for mutable improvers \citep{Zhang2025DGM,Liu2026RSISurvey,Duan2026LastAI}. Their practical effectiveness remains a separate empirical question.

The controlled experiments test the distinctions within a declared finite family. Low distributional risk can justify a different stopping decision from exact completion. Persistent query-program revision recovers coverage after a distribution shift, but current-sample greedy selection performs better at lower cost. Neither the experiments nor fixed-layer closure establish a rate limit on future self-improvement.

Section~\ref{sec:resource-contract} states the resource and answer contracts. Section~\ref{sec:acquisition} gives certified acquisition, and Section~\ref{sec:certification-barrier} identifies the obstruction to uniformly supplying its completion specifications. Section~\ref{sec:experiments} reports the controlled evidence. The appendices contain the proofs, full computation model, classical limiting results, assistance applications, and experimental protocols.

\section{Resource and answer contracts}
\label{sec:resource-contract}
A decision task is a set $D\subseteq\mathbb N$. Write $D\leq_T A$ when a total oracle procedure with access to $A$ decides membership in $D$, and define
\[
\mathcal C(A)=\{D:D\leq_T A\}.
\]
The class places no runtime or memory bound on a terminating computation. It records resource-relative solvability, not the efficiency or repertoire of a particular implementation. Additional answering resources are included through a join $A\oplus E$. Write $A'$ for the relative diagonal halting set; the classical jump theorem gives $A'\not\leq_T A$.

The complete mutable state includes encoded solvers, stored data, trainers, evaluators, interpreters, and improvers. Uniformly effective revision over $A$ leaves every total resulting decision procedure in $\mathcal C(A)$, even when the improver itself changes (Theorem~\ref{thm:simulation}). Uniform use of additional resources is bounded by their join (Theorem~\ref{thm:resources}). A fixed finite training artifact can be hardcoded and cannot alone supply strict degree ascent (Corollary~\ref{cor:finite}). Appendix~\ref{sec:model} gives the full digital-agent and LLM-service semantics.

An acquisition claim also needs an answer contract. In the model below, a teacher supplies the exact bit of one fixed unknown target, the hypothesis family is specified in advance, and initial resources are target independent. Learned proposals may choose queries or suggest procedures; they do not establish the teacher's correctness. The global feature-span or separating specification is a separate structural premise. The resulting all-input guarantee is conditional on these declared contracts.

\section{Certified acquisition within a fixed layer}
\label{sec:acquisition}

This section studies an exact instance of learning a rule from teacher queries and retaining it after those queries cease. Membership-query learning provides the classical setting \citep{Angluin1988Queries}. The restriction to a known linear feature language permits complete proofs about acquisition and stopping. It also exposes a separate difficulty: obtaining the semantic specification that makes the stopping test valid.

\subsection{Task family, information, and mutable proposer}

Fix a baseline $A$, an integer $d\geq1$, and a known total $A$-computable feature map $\phi:\N\longrightarrow\mathbb F_2^d$. The target family is specified in advance by
\[
\mathcal H_\phi=\{f_\theta:\theta\in\mathbb F_2^d\},
\qquad f_\theta(x)=\theta^\top\phi(x)\pmod 2.
\]
A teacher returns the exact membership bit $f_\theta(x)$ for one fixed but unknown target. The learner's initial state, $A$, and $\phi$ are the same for all targets; there is no target-dependent advice other than the teacher's replies. Write $D_\theta=\{x:f_\theta(x)=1\}$. Each fixed $D_\theta\T A$: its finite coefficient vector can be hardcoded. This does not give a learner that starts without $\theta$ the information needed to select that decider.

Let $W=\operatorname{span}\{\phi(x):x\in\N\}$ and $r=\dim W$. For the positive protocol, the task specification supplies inputs $z_1,\ldots,z_r$ whose feature vectors are a basis of $W$, together with the structural guarantee that every feature vector lies in their span. The guarantee is a premise about the feature language, not an inference from finitely many observed inputs. For an explicit feature language it may have a direct proof. Section~\ref{sec:span-specification} examines the difficulty of deriving it from arbitrary code.

After a finite query history $T=((x_1,y_1),\ldots,(x_m,y_m))$, define
\[
V_T=\operatorname{span}\{\phi(x_1),\ldots,\phi(x_m)\},\qquad
\Theta_T=\{\vartheta:\vartheta^\top\phi(x_j)=y_j\text{ for every }j\}.
\]
Realizability and noiseless replies ensure $\Theta_T\ne\varnothing$. The state stores these constraints, a consistent solver, and the code of a query proposer. The proposer and its revision rule may change effectively using $A$ and $T$. Teacher calls are made only through the protocol below. A protected constraint store and rank check govern acceptance; arbitrary changes to these components are outside this restricted model.

\subsection{Exact answer identification and finite completion}

\begin{theorem}[Query span, completion, and teacher removal]
\label{thm:query-span}
For the task family and information model above:
\begin{enumerate}
\item The history $T$ uniquely determines the target answer on $x$, over all targets consistent with $T$, if and only if $\phi(x)\in V_T$.
\item A history certifies a correct teacher-free solver on every input if and only if $V_T=W$. When this holds, any $\widehat\theta\in\Theta_T$ defines such a solver.
\item There is a terminating protocol with mutable query proposers using at most $r$ teacher calls. Its worst-case query count is optimal among deterministic learners that identify every target in $\mathcal H_\phi$ from the stated initial information.
\item Its resulting solver decides $D_\theta$ relative to $A$ without teacher access. In particular, $\deg_T(A\oplus D_\theta)=\deg_T(A)$ throughout this family.
\end{enumerate}
\end{theorem}

The proof is in Appendix~\ref{app:acquisition-proofs}.

The protected protocol first tests completion. Otherwise it accepts a proposer query only after a bounded check establishes a new feature direction; a failed or redundant proposal triggers a supplied-basis fallback. Every teacher call increases rank, and bounded proposal execution prevents an unproductive proposer from blocking termination. The full protocol and lower-bound proof are in Appendix~\ref{app:acquisition-proofs}.

The first two assertions concern a guarantee valid for every consistent target. A learner can guess correctly before spanning $W$, but the history cannot certify that guess under this information model. The lower bound likewise concerns uniform identification from unknown target information, whereas $D_\theta\T A$ concerns the existence of a decider for each fixed target. The two quantifier orders express different questions.

The step allowances, feature evaluations, and linear-system operations contribute to computation cost; the theorem bounds teacher queries, not wall-clock time. It also does not certify an arbitrary replacement solver from its code: deployment uses the specified feature evaluator and a coefficient vector checked against the constraints. Revising the feature language requires re-establishing the task-family and span premises.

\begin{example}[Distinct queries can leave the same unresolved direction]
\label{ex:parity-acquisition}
Take $A=\varnothing$, $d=3$, and let $\phi(n)$ be the three low-order bits of $n$, listed from least to most significant. The inputs $1,2,4$ supply a basis, and their span contains every feature vector by construction. For target $\theta=(1,1,0)$, the queries $3,5,6$ return $0,1,1$. These three distinct queries have rank two: $\phi(6)=\phi(3)+\phi(5)$. The alternative target $(0,0,1)$ gives the same three replies, but disagrees on input $1$. Agreement on these tests therefore cannot certify teacher removal.

After the first two queries, the protocol rejects $6$ as redundant and can use the basis input $1$. Its third teacher call then gives rank three and a solver correct on all $n\in\N$, with no further teacher calls. The conclusion extends beyond the queried inputs because of the proved feature specification. No degree ascent occurs.
\end{example}

Figure~\ref{fig:query-span} contrasts these two query histories. The completion certificate uses both the acquired constraints and the independently specified feature span.

\begin{figure}[htbp]
\centering
\begin{tikzpicture}[
font=\small,>=Stealth,
query/.style={draw,rounded corners=3pt,fill=blue!5,align=center,
text width=2.8cm,minimum height=1.7cm,inner sep=6pt},
outcome/.style={draw,rounded corners=3pt,align=center,
text width=5cm,minimum height=1.75cm,inner sep=6pt},
every path/.style={line width=.65pt}]
\node[align=center,text width=14.8cm] at (0,2.7)
{Known feature language: three low-order bits, least significant first\\
Supplied span: $W=\mathbb F_2^3$};
\node[query] (qthree) at (-5.55,0)
{\textbf{Query 3}\\[3pt]
$\phi(3)=(1,1,0)$\\$y=0$\\[3pt]$\dim V_T=1$};
\node[query] (qfive) at (-1.85,0)
{\textbf{Query 5}\\[3pt]
$\phi(5)=(1,0,1)$\\$y=1$\\[3pt]$\dim V_T=2$};
\node[outcome,fill=orange!8] (redundant) at (4.35,1.2)
{\textbf{Redundant third query: 6}\\[3pt]
$\phi(6)=(0,1,1),\quad y=1$\\
$\dim V_T=2<3$\\[3pt]
Teacher removal uncertified};
\node[outcome,fill=green!6] (complete) at (4.35,-1.2)
{\textbf{Guarded third query: 1}\\[3pt]
$\phi(1)=(1,0,0),\quad y=1$\\
$\dim V_T=3=\dim W$\\[3pt]
Teacher can be removed};
\draw[->] (qthree.east) -- (qfive.west);
\draw[->,dashed] (qfive.east) -- (.55,0) |- (redundant.west);
\draw[->] (qfive.east) -- (.55,0) |- (complete.west);
\node[draw,rounded corners=3pt,fill=gray!6,align=center,anchor=north,
text width=14cm,inner sep=7pt] at (0,-2.75)
{Completion: $V_T=W\ \Longrightarrow\ f_{\widehat\theta}(x)=f_\theta(x)$ for all $x$, with $\widehat\theta\in\Theta_T$.\\[3pt]
The learned task remains computable: $D_\theta\T\varnothing$.};
\end{tikzpicture}
\caption{Certified teacher removal in Example~\ref{ex:parity-acquisition}, illustrated for the unknown target $\theta=(1,1,0)$. The dashed branch gives three distinct queries whose feature vectors still have rank two: targets $(1,1,0)$ and $(0,0,1)$ agree on all three replies. The guarded branch rejects query $6$ before calling the teacher and queries $1$ instead. Its third reply completes the known span, certifying every future answer under the stated family and noiseless-teacher assumptions. Teacher removal changes the learner's available capability without changing the Turing degree of the resource base.}
\label{fig:query-span}
\end{figure}

The proposer can still matter before completion. For a declared task distribution $\mu$, the quantity $\Pr_{X\sim\mu}[\phi(X)\in V_T]$ measures the mass of inputs whose answers are already certified. Query ordering can change this coverage and the cost of obtaining it. A self-revised proposer may improve either quantity; such improvement must be assessed separately and is not guaranteed by permitting revision. The rank guard supplies a guarantee that remains valid across these changes.

\subsection{Finite separating specifications beyond linear features}
\label{sec:finite-separators}

The linear feature span is one instance of a more general completion
specification. Let $\mathcal H=\{h_1,\ldots,h_N\}$ be a supplied nonempty
finite family of total Boolean functions, evaluated uniformly relative to
$A$. Distinct descriptions may compute the same function; $N$ counts the
supplied descriptions. A finite set $Q\subseteq\N$ is a \emph{separating
specification} for this family if
\[
\forall h,k\in\mathcal H,\quad
\bigl[\forall q\in Q\ h(q)=k(q)\bigr]
\Longrightarrow \bigl[\forall x\ h(x)=k(x)\bigr].
\]
This is structural information about the supplied programs. It is stronger
than observing agreement on $Q$ for one pair of programs.

\begin{theorem}[Certified finite-family acquisition]
\label{thm:finite-acquisition}
Suppose $Q$ is a supplied separating specification and the teacher returns
exact labels from an unknown $h_*\in\mathcal H$. Maintain the nonempty
version space $V\subseteq\mathcal H$ consistent with all replies. If its
members disagree on some $q\in Q$, query such a $q$ and retain precisely
the members agreeing with the reply. Otherwise return any member of $V$.
This procedure terminates after at most $N-1$ teacher calls and returns a
solver equal to $h_*$ on every input. The returned decision task remains
in $\mathcal C(A)$.
\end{theorem}

The proof is in Appendix~\ref{app:acquisition-proofs}.

The elimination argument is classical exact query learning
\citep{Angluin1988Queries}. Its role here is to exhibit the structural
contract that certifies teacher removal beyond the linear family. The
$N-1$ guarantee is a general upper bound, not an improvement over the
rank bound in Theorem~\ref{thm:query-span}: a rank-$r$ linear family has
$2^r$ distinct functions but needs only $r$ independent labels.
An executable Lean development proves termination, the query bound,
retained membership, and all-input agreement for this finite-family
procedure (Appendix~\ref{app:formalization}).

\section{Why completion specifications cannot always be supplied automatically}
\label{sec:certification-barrier}
The positive contracts separate acquisition from certification of completion. The following results identify why the structural premise cannot be generated uniformly for arbitrary supplied programs. They concern semantic certification, even when every individual target is baseline-computable.

\subsection{The difficulty of supplying the span specification}
\label{sec:span-specification}

The positive result assumes a valid description of $W$. Finite dimension alone does not provide an effective way to obtain one from arbitrary feature code.

\begin{proposition}[Uniform span specification requires stronger information]
\label{prop:span-hardness}
No $A$-computable procedure, given an index for an arbitrary total $A$-computable map $\phi:\N\to\mathbb F_2$, always determines $\dim\operatorname{span}(\operatorname{range}\phi)$. Conversely, $A'$ suffices uniformly to determine the feature span and produce basis inputs for every total $A$-computable $\phi:\N\to\mathbb F_2^d$ with finite specified $d$.
\end{proposition}

The proof is in Appendix~\ref{app:acquisition-proofs}.

This is a halting reduction, not a claim that ordinary feature languages require a jump oracle. The low-bit language above has an elementary span proof. The proposition identifies why no general baseline procedure can supply that proof's semantic conclusion for every effective feature program. Obtaining a finite specification uniformly can require stronger information even though using any one correct specification is an effective finite process.

The model therefore separates three claims: query replies identify a target inside a known family; a structural specification makes finite completion recognizable; neither operation gives the learned target a higher degree. The relative jump enters at the level of uniform semantic certification, rather than as a purported capability of the learned solver. This yields a concrete application of the distinction between finite artifacts and the procedures that select or certify them.

Supplying the family as total programs does not by itself supply $Q$.
The following result locates the missing semantic information precisely.

\begin{proposition}[Uniform separator specification]
\label{prop:separator-hardness}
There is no $A$-computable procedure that, from indices for every promised
nonempty finite family of total $A$-computable Boolean functions, always
returns a finite separating specification. This already fails for two
supplied descriptions. Conversely, $A'$ suffices uniformly to produce a
separating specification for every such family.
\end{proposition}

The proof is in Appendix~\ref{app:acquisition-proofs}.

The possibility of equivalent program descriptions is essential to the
lower bound. If every pair is instead promised to compute distinct
functions, a baseline search for a disagreement terminates for each of
the finitely many pairs and constructs $Q$ without $A'$. The obstruction
is to completing this construction uniformly when an equivalent pair may
have no disagreement to find. It does not apply to a supplied finite
family whose pairwise distinctions or separating specification are
already justified.

Together, these results distinguish the information needed to obtain a
completion specification from the information needed to use it. A valid
finite specification supports effective acquisition and deployment in the
baseline layer. Producing such specifications uniformly from arbitrary
programs has the relative-halting barrier even when every candidate is
individually total and baseline-computable. The Lean artifact checks the
finite-list reduction underlying the lower bound; the classical relative
jump theorem and the $A'$ upper-bound construction remain mathematical
premises and proofs outside that artifact.

\section{Controlled acquisition experiments}
\label{sec:experiments}

The model in Section~\ref{sec:acquisition} separates three operational questions: whether a solver has low error, whether its answers are certified on every input, and whether revising its query program improves the cost of acquisition. We evaluate these distinctions in a controlled Boolean family. The experiments supplement the formal analysis with measured tradeoffs; the query-learning algorithms and confidence bound are elementary baselines.

\subsection{Evaluation questions}
The experiments use exact Boolean teachers, a known coordinate basis, and input distributions mixing a frequent subspace with the full feature space. We evaluate distributional risk and the mass of inputs whose answers are certified separately. Full task generation, exact evaluation formulas, and policy checks are in Appendix~\ref{app:acquisition-evaluation}.

\subsection{Stopping with different guarantees}

Six stopping rules compare fixed acquisition budgets, prediction stability, fresh-label risk validation, rank completion on a stream, and active queries to the supplied basis. The study uses dimensions $8,12,16$, 200 independent tasks per condition, a 2\% risk target, and a 5\% statistical family error allowance. Structural completion guarantees every input; statistical validation bounds risk under the declared distribution. The complete stopping rules, label costs, and results are in Appendix~\ref{app:acquisition-evaluation}.

\begin{table}[htbp]
\centering\small
\begin{tabular}{@{}crrr@{}}
\toprule
Full-space mass $\rho$ & Statistical risk & Rank on stream & Supplied basis\\
\midrule
$1$ & $100.00$ & $100.00$ & $100.00$\\
$.1$ & $88.00$ & $76.00$ & $100.00$\\
$.01$ & $90.67$ & $4.83$ & $100.00$\\
\bottomrule
\end{tabular}
\caption{Teacher-removal rates (\%) under different completion criteria, averaged over dimensions $8,12,16$ with 600 tasks per row. Statistical removal targets at most 2\% distributional risk; the two rank-based methods target exact identification. Label costs differ and are reported in Appendix~\ref{app:acquisition-evaluation}.}
\label{tab:completion-summary}
\end{table}

Table~\ref{tab:completion-summary} shows why the guarantees must be distinguished; Appendix Table~\ref{tab:stopping} gives the corresponding risks and label costs. At $\rho=.1$, prediction stability removes the teacher with excessive risk in 193 of 600 trials (32.17\%; Wilson 95\% interval 28.55--36.01\%). At $\rho=.01$, however, a solver can correctly satisfy the 2\% risk requirement while remaining wrong on some inputs. The statistical rule removes the teacher in 90.67\% of trials, with mean conditional risk 0.335\% and no observed threshold violation. Rank-on-stream removes it in only 4.83\% of trials within the same acquisition cap because rare directions remain unobserved. This difference is consistent with the distinct objectives of distributional accuracy and exact identification.

The supplied basis achieves exact completion in $d$ labels in every trial. Its advantage depends on the known global feature specification and the ability to select informative inputs. The validation-only statistical baseline does not exploit that specification or recycle validation labels into acquisition constraints; its large validation cost is not a lower bound for statistical learning. These results illustrate the value of structural information when it is available, while Proposition~\ref{prop:span-hardness} concerns the separate difficulty of obtaining it from arbitrary code.

\subsection{Persistent query-program revision and its cost}

A second study compares fixed, inherited-and-revised, and current-sample greedy query programs across a distribution shift. Each program uses the same protected rank guard, makes exactly $d$ teacher calls, and ends with zero risk. Forty independent replicates at each of $d=8,12$ contain six episodes, with the frequent subspace changing after episode three. The measured difference is pre-completion certified coverage and its computational cost. Mutation operators, shared samples, and budget matching are specified in Appendix~\ref{app:acquisition-evaluation}.

\begin{figure}[htbp]
\centering
\begin{tikzpicture}[font=\footnotesize]
\definecolor{emp0}{HTML}{777777}
\definecolor{emp1}{HTML}{B69B4C}
\definecolor{emp2}{HTML}{7B5BA7}
\definecolor{emp3}{HTML}{007E79}
\definecolor{emp4}{HTML}{CB482A}
\definecolor{emp5}{HTML}{3C70B0}
\begin{scope}[xshift=0cm,x=.51cm,y=3.6cm]
\foreach \y in {0,.2,.4,.6,.8,1}{\draw[gray!18] (0,\y)--(12,\y);}
\draw (0,1.02)--(0,0)--(12.1,0);
\foreach \x in {0,2,4,6,8,10,12}{\draw (\x,0)--(\x,-.018) node[below]{\x};}
\node at (6,-.18) {Teacher queries};
\node[font=\small] at (6,1.1) {Before shift};
\foreach \y in {0,.2,.4,.6,.8,1}{\node[left] at (0,\y){\y};}
\node[rotate=90] at (-1.5,.5) {Exact certified coverage};
\fill[emp0,opacity=.10] plot coordinates {(0,0.0140869) (1,0.0141113) (2,0.0141602) (3,0.0153125) (4,0.0179687) (5,0.0194141) (6,0.0254688) (7,0.0396875) (8,0.0695313) (9,0.1306250) (10,0.2500000) (11,0.5000000) (12,1.0000000) (12,1.0000000) (11,0.5000000) (10,0.2668750) (9,0.1559375) (8,0.0906250) (7,0.0530469) (6,0.0339063) (5,0.0264453) (4,0.0232422) (3,0.0184766) (2,0.0166211) (1,0.0158691) (0,0.0140869)} -- cycle;
\draw[emp0,line width=.8pt] plot coordinates {(0,0.0140869) (1,0.0148145) (2,0.0152148) (3,0.0167188) (4,0.0204297) (5,0.0225781) (6,0.0293359) (7,0.0460156) (8,0.0793750) (9,0.1418750) (10,0.2556250) (11,0.5000000) (12,1.0000000)};
\fill[emp1,opacity=.10] plot coordinates {(0,0.0140869) (1,0.0141113) (2,0.0143945) (3,0.0150781) (4,0.0163281) (5,0.0204688) (6,0.0279297) (7,0.0417969) (8,0.0709375) (9,0.1325000) (10,0.2556250) (11,0.5037500) (12,1.0000000) (12,1.0000000) (11,0.5300000) (10,0.2762500) (9,0.1446875) (8,0.0784375) (7,0.0474219) (6,0.0323828) (5,0.0236328) (4,0.0182031) (3,0.0166016) (2,0.0154492) (1,0.0146973) (0,0.0140869)} -- cycle;
\draw[emp1,line width=.8pt,densely dashed] plot coordinates {(0,0.0140869) (1,0.0143457) (2,0.0148633) (3,0.0157813) (4,0.0172656) (5,0.0219922) (6,0.0300391) (7,0.0443750) (8,0.0746875) (9,0.1381250) (10,0.2650000) (11,0.5150000) (12,1.0000000)};
\fill[emp2,opacity=.10] plot coordinates {(0,0.0140869) (1,0.0281738) (2,0.0563477) (3,0.1126953) (4,0.2253906) (5,0.4507813) (6,0.9015625) (7,0.9031250) (8,0.9062500) (9,0.9125000) (10,0.9250000) (11,0.9500000) (12,1.0000000) (12,1.0000000) (11,0.9500000) (10,0.9250000) (9,0.9125000) (8,0.9062500) (7,0.9031250) (6,0.9015625) (5,0.4507813) (4,0.2253906) (3,0.1126953) (2,0.0563477) (1,0.0281738) (0,0.0140869)} -- cycle;
\draw[emp2,line width=.8pt] plot coordinates {(0,0.0140869) (1,0.0281738) (2,0.0563477) (3,0.1126953) (4,0.2253906) (5,0.4507813) (6,0.9015625) (7,0.9031250) (8,0.9062500) (9,0.9125000) (10,0.9250000) (11,0.9500000) (12,1.0000000)};
\fill[emp3,opacity=.10] plot coordinates {(0,0.0140869) (1,0.0281738) (2,0.0563477) (3,0.1126953) (4,0.2253906) (5,0.4507813) (6,0.9015625) (7,0.9031250) (8,0.9062500) (9,0.9125000) (10,0.9250000) (11,0.9500000) (12,1.0000000) (12,1.0000000) (11,0.9500000) (10,0.9250000) (9,0.9125000) (8,0.9062500) (7,0.9031250) (6,0.9015625) (5,0.4507813) (4,0.2253906) (3,0.1126953) (2,0.0563477) (1,0.0281738) (0,0.0140869)} -- cycle;
\draw[emp3,line width=.8pt] plot coordinates {(0,0.0140869) (1,0.0281738) (2,0.0563477) (3,0.1126953) (4,0.2253906) (5,0.4507813) (6,0.9015625) (7,0.9031250) (8,0.9062500) (9,0.9125000) (10,0.9250000) (11,0.9500000) (12,1.0000000)};
\fill[emp4,opacity=.10] plot coordinates {(0,0.0140869) (1,0.0281738) (2,0.0563477) (3,0.1126953) (4,0.2253906) (5,0.4507813) (6,0.9015625) (7,0.9031250) (8,0.9062500) (9,0.9125000) (10,0.9250000) (11,0.9500000) (12,1.0000000) (12,1.0000000) (11,0.9500000) (10,0.9250000) (9,0.9125000) (8,0.9062500) (7,0.9031250) (6,0.9015625) (5,0.4507813) (4,0.2253906) (3,0.1126953) (2,0.0563477) (1,0.0281738) (0,0.0140869)} -- cycle;
\draw[emp4,line width=.8pt] plot coordinates {(0,0.0140869) (1,0.0281738) (2,0.0563477) (3,0.1126953) (4,0.2253906) (5,0.4507813) (6,0.9015625) (7,0.9031250) (8,0.9062500) (9,0.9125000) (10,0.9250000) (11,0.9500000) (12,1.0000000)};
\fill[emp5,opacity=.10] plot coordinates {(0,0.0140869) (1,0.0271191) (2,0.0511914) (3,0.0948828) (4,0.1742969) (5,0.3045313) (6,0.4750000) (7,0.6368750) (8,0.7206250) (9,0.7868750) (10,0.8575000) (11,0.9237500) (12,1.0000000) (12,1.0000000) (11,0.9500000) (10,0.9025000) (9,0.8431250) (8,0.7956250) (7,0.7278359) (6,0.5640859) (5,0.3485938) (4,0.1949219) (3,0.1047266) (2,0.0545898) (1,0.0280566) (0,0.0140869)} -- cycle;
\draw[emp5,line width=.8pt,densely dashed] plot coordinates {(0,0.0140869) (1,0.0275879) (2,0.0529492) (3,0.1000391) (4,0.1850781) (5,0.3270312) (6,0.5190625) (7,0.6837500) (8,0.7600000) (9,0.8150000) (10,0.8800000) (11,0.9387500) (12,1.0000000)};
\end{scope}
\begin{scope}[xshift=7.35cm,x=.51cm,y=3.6cm]
\foreach \y in {0,.2,.4,.6,.8,1}{\draw[gray!18] (0,\y)--(12,\y);}
\draw (0,1.02)--(0,0)--(12.1,0);
\foreach \x in {0,2,4,6,8,10,12}{\draw (\x,0)--(\x,-.018) node[below]{\x};}
\node at (6,-.18) {Teacher queries};
\node[font=\small] at (6,1.1) {After shift};
\fill[emp0,opacity=.10] plot coordinates {(0,0.0140869) (1,0.0141113) (2,0.0141602) (3,0.0142578) (4,0.0148047) (5,0.0173047) (6,0.0230078) (7,0.0368750) (8,0.0681250) (9,0.1250000) (10,0.2500000) (11,0.5000000) (12,1.0000000) (12,1.0000000) (11,0.5000000) (10,0.2668750) (9,0.1446875) (8,0.0807812) (7,0.0474219) (6,0.0310938) (5,0.0215234) (4,0.0176172) (3,0.0160156) (2,0.0152148) (1,0.0141113) (0,0.0140869)} -- cycle;
\draw[emp0,line width=.8pt] plot coordinates {(0,0.0140869) (1,0.0141113) (2,0.0145117) (3,0.0149609) (4,0.0162109) (5,0.0194141) (6,0.0268750) (7,0.0417969) (8,0.0737500) (9,0.1334375) (10,0.2556250) (11,0.5000000) (12,1.0000000)};
\fill[emp1,opacity=.10] plot coordinates {(0,0.0140869) (1,0.0141113) (2,0.0141602) (3,0.0146064) (4,0.0163252) (5,0.0187109) (6,0.0272266) (7,0.0401563) (8,0.0676563) (9,0.1287500) (10,0.2500000) (11,0.5000000) (12,1.0000000) (12,1.0000000) (11,0.5000000) (10,0.2631250) (9,0.1381250) (8,0.0742188) (7,0.0467187) (6,0.0314453) (5,0.0218750) (4,0.0189062) (3,0.0156641) (2,0.0145117) (1,0.0141113) (0,0.0140869)} -- cycle;
\draw[emp1,line width=.8pt,densely dashed] plot coordinates {(0,0.0140869) (1,0.0141113) (2,0.0142773) (3,0.0150781) (4,0.0175000) (5,0.0202344) (6,0.0292187) (7,0.0432031) (8,0.0709375) (9,0.1334375) (10,0.2556250) (11,0.5000000) (12,1.0000000)};
\fill[emp2,opacity=.10] plot coordinates {(0,0.0140869) (1,0.0141113) (2,0.0141602) (3,0.0142578) (4,0.0151563) (5,0.0176563) (6,0.0240625) (7,0.0368750) (8,0.0639062) (9,0.1278125) (10,0.2500000) (11,0.5000000) (12,1.0000000) (12,1.0000000) (11,0.5000000) (10,0.2781250) (9,0.1475000) (8,0.0751563) (7,0.0453125) (6,0.0289844) (5,0.0218750) (4,0.0183203) (3,0.0160156) (2,0.0152148) (1,0.0141113) (0,0.0140869)} -- cycle;
\draw[emp2,line width=.8pt] plot coordinates {(0,0.0140869) (1,0.0141113) (2,0.0145117) (3,0.0149609) (4,0.0165625) (5,0.0197656) (6,0.0265234) (7,0.0410937) (8,0.0695312) (9,0.1362500) (10,0.2612500) (11,0.5000000) (12,1.0000000)};
\fill[emp3,opacity=.10] plot coordinates {(0,0.0140869) (1,0.0281738) (2,0.0563477) (3,0.1126953) (4,0.2253906) (5,0.4507813) (6,0.9015625) (7,0.9031250) (8,0.9062500) (9,0.9125000) (10,0.9250000) (11,0.9500000) (12,1.0000000) (12,1.0000000) (11,0.9500000) (10,0.9250000) (9,0.9125000) (8,0.9062500) (7,0.9031250) (6,0.9015625) (5,0.4507813) (4,0.2253906) (3,0.1126953) (2,0.0563477) (1,0.0281738) (0,0.0140869)} -- cycle;
\draw[emp3,line width=.8pt] plot coordinates {(0,0.0140869) (1,0.0281738) (2,0.0563477) (3,0.1126953) (4,0.2253906) (5,0.4507813) (6,0.9015625) (7,0.9031250) (8,0.9062500) (9,0.9125000) (10,0.9250000) (11,0.9500000) (12,1.0000000)};
\fill[emp4,opacity=.10] plot coordinates {(0,0.0140869) (1,0.0278223) (2,0.0544727) (3,0.1042578) (4,0.1967969) (5,0.3757813) (6,0.6803125) (7,0.7962031) (8,0.8425000) (9,0.8787500) (10,0.9062500) (11,0.9500000) (12,1.0000000) (12,1.0000000) (11,0.9500000) (10,0.9250000) (9,0.9087500) (8,0.8837500) (7,0.8506250) (6,0.7496875) (5,0.4123438) (4,0.2136719) (3,0.1103516) (2,0.0561133) (1,0.0281738) (0,0.0140869)} -- cycle;
\draw[emp4,line width=.8pt] plot coordinates {(0,0.0140869) (1,0.0280566) (2,0.0554102) (3,0.1075391) (4,0.2057031) (5,0.3935938) (6,0.7168750) (7,0.8225000) (8,0.8631250) (9,0.8937500) (10,0.9175000) (11,0.9500000) (12,1.0000000)};
\fill[emp5,opacity=.10] plot coordinates {(0,0.0140869) (1,0.0272363) (2,0.0518945) (3,0.0944023) (4,0.1696094) (5,0.3035937) (6,0.5002891) (7,0.6612500) (8,0.7543750) (9,0.8243281) (10,0.8875000) (11,0.9312500) (12,1.0000000) (12,1.0000000) (11,0.9500000) (10,0.9212500) (9,0.8750000) (8,0.8331250) (7,0.7550234) (6,0.6034375) (5,0.3500117) (4,0.1907031) (3,0.1033203) (2,0.0551758) (1,0.0280566) (0,0.0140869)} -- cycle;
\draw[emp5,line width=.8pt,densely dashed] plot coordinates {(0,0.0140869) (1,0.0277051) (2,0.0535352) (3,0.0991016) (4,0.1799219) (5,0.3265625) (6,0.5514063) (7,0.7100000) (8,0.7956250) (9,0.8506250) (10,0.9062500) (11,0.9425000) (12,1.0000000)};
\end{scope}
\draw[emp0,line width=.8pt] (0.00,-1.25)--(0.45,-1.25) node[right,black] {Fixed basis};
\draw[emp1,line width=.8pt,densely dashed] (4.70,-1.25)--(5.15,-1.25) node[right,black] {Random program};
\draw[emp2,line width=.8pt] (9.40,-1.25)--(9.85,-1.25) node[right,black] {Development frozen};
\draw[emp3,line width=.8pt] (0.00,-1.67)--(0.45,-1.67) node[right,black] {Current greedy};
\draw[emp4,line width=.8pt] (4.70,-1.67)--(5.15,-1.67) node[right,black] {Revised program};
\draw[emp5,line width=.8pt,densely dashed] (9.40,-1.67)--(9.85,-1.67) node[right,black] {Search from basis};
\end{tikzpicture}
\caption{Certified coverage during acquisition for $d=12$. Curves average 40 independent replicates, each containing three episodes before and three after a change in the frequent subspace. Shading gives pointwise 95\% bootstrap intervals over whole replicates. Before the shift, the frozen, revised, and current-greedy curves coincide; after it, revision recovers much of the coverage lost by the frozen program, while greedy performs better at lower computational cost. Every method reaches coverage one after 12 teacher queries.}
\label{fig:empirical-coverage}
\end{figure}

Define mean pre-completion coverage as $\bar c=d^{-1}\sum_{q=0}^{d-1}c(V_q)$. After the shift at $d=12$, the revised program attains $\bar c=.4973$, compared with $.0941$ for the frozen program and $.4548$ for search restarted from the basis. The paired gains are $.4033$ (95\% bootstrap interval $[.3936,.4123]$) and $.0426$ ($[.0275,.0584]$), respectively. Inheritance retains previous computation as well as a program, so the latter comparison matches online search budgets, not total historical work.

The stronger greedy baseline attains $.5322$, exceeding revision by $.0348$ ($[.0261,.0441]$), while using 768 sample visits and approximately 1,581 elimination XOR operations per episode; revision uses 49,920 visits and 20,048 XOR operations. Mean measured times are 0.50 and 4.67 milliseconds in this implementation. Frozen and revised programs additionally require the one-time development search, approximately 6,531 XOR operations at $d=12$. The same ordering holds at $d=8$: post-shift coverage is $.5679$ for greedy and $.5570$ for revision, with a paired difference of $.0109$ ($[.0054,.0175]$).

For these distributions, a rank-$q$ span satisfies
\[
c(V_q)\leq(1-\rho)2^{\min(q,\dim H)-\dim H}+\rho2^{q-d},
\]
because $\dim(V_q\cap H)\leq\min(q,\dim H)$. Current-sample greedy attains this envelope at every measured query budget in all 480 episodes. Before the shift, the revised, frozen, and greedy programs already have identical exact coverage, even though revision accepts empirical-score improvements. Thus accepted program edits need not increase true coverage. Figure~\ref{fig:empirical-coverage} shows both the recovery after a shift and the remaining gap to greedy.

These findings demonstrate adaptive query ordering within a fixed computability layer, with unchanged exact query complexity. They supply no evidence that program revision is superior to a suitable classical policy, that the improver improves itself, or that the effect transfers to unrestricted agent workloads. The empirical contribution is the separation of stopping guarantees, interim acquisition benefits, and revision costs under one explicitly specified model.

\section{Discussion and conclusion}
Certified teacher removal requires both observations that identify the target and structural information establishing that the observations cover the task family. The acquisition protocols exhibit this distinction constructively: a learner can revise its query proposer, obtain a finite deployed solver, and stop calling its teacher with an all-input guarantee. The span and separator reductions locate a complementary obstruction to uniformly extracting the completion specification from arbitrary total code.

These results concern exact acquisition under supplied contracts. Membership-query learning, learning-to-learn, and generalized distillation address related acquisition questions with different assumptions \citep{Angluin1988Queries,Baxter2000Bias,LopezPaz2015Distillation}. Our linear-algebraic and elimination arguments are classical; their role is to expose the acquisition, completion, and resource obligations in one evolving-agent model. An imperfect teacher, a changed representation, or arbitrary replacement code requires additional justification. Exact completion and low statistical risk remain distinct objectives.

The computability boundary preserves substantial practical improvement. Program revision can change query order, efficiency, and reuse while the complete process remains in $\mathcal C(A)$. The controlled experiments show adaptation after a distribution shift, together with a stronger greedy baseline and its lower cost. They establish neither a general advantage of self-revision nor a trajectory through Turing degrees. Appendix~\ref{sec:cases} discusses contemporary architectures, and Appendix~\ref{sec:empirical} separates measured support removal from degree attribution.

Construction assistance, acquisition feedback, and deployment calls must be counted at their respective boundaries. Removing a teacher at deployment does not establish that its information played no role in construction. Human input is necessary under an explicit premise that a human is the exclusive source of required distinguishing task information; neither this premise nor universal human necessity follows from halting undecidability (Appendix~\ref{sec:assistance}). The central conclusion is a precise account of what must be supplied and certified to turn a supported acquisition process into an independently deployed capability.

\appendix
\section{Proofs of acquisition and completion-specification results}
\label{app:acquisition-proofs}

\begin{proof}[Proof of Theorem~\ref{thm:query-span}]
If $\phi(x)=\sum_j c_j\phi(x_j)$, every consistent target gives the answer $\sum_j c_j y_j$. Conversely, if $\phi(x)\notin V_T$, finite-dimensional linear algebra supplies a vector $u$ orthogonal to $V_T$ with $u^\top\phi(x)=1$. For any $\vartheta\in\Theta_T$, both $\vartheta$ and $\vartheta+u$ are consistent with the history and disagree on $x$. This proves the first assertion. Since $W$ is generated by the entire feature image, all answers are identified exactly when $V_T=W$, proving the second.

For the protocol, first test whether $V_T=W$ by comparing ranks with the supplied basis. If not, run the current proposer for a prescribed finite, $A$-computable step allowance. Use its proposed input only if a bounded check produces a feature vector outside $V_T$; otherwise choose a supplied basis input whose vector lies outside $V_T$. Query the teacher once, add the constraint, and permit effective revisions to the proposer within a finite step allowance. Feature computation on the fallback inputs terminates by totality. Every teacher call increases the rank by one, so the protocol terminates after at most $r$ calls, including the case $r=0$ with no calls. Finite allowances prevent a divergent or repeatedly uninformative proposer from obstructing completion.

There are $2^r$ distinct target functions: restrictions of coefficient vectors to $W$ realize all linear functionals on $W$. A deterministic learner with worst-case depth $q$ has at most $2^q$ binary-answer leaves. Two distinct target functions cannot lead to the same terminal solver. Hence $q\geq r$, regardless of how the learner revises itself between queries. This proves optimality. Finally, solving the finite linear system yields a finite $\widehat\theta$, and computing $\widehat\theta^\top\phi(x)$ uses only $A$. Since $D_\theta\T A$, adjoining it preserves the degree.
\end{proof}

\begin{proof}[Proof of Theorem~\ref{thm:finite-acquisition}]
The true target is retained after every exact reply. At a disagreement,
at least one candidate predicts each Boolean answer, so either reply
removes at least one member while preserving the target. Thus at most
$N-1$ queries occur. When no disagreement remains on $Q$, every survivor
agrees with $h_*$ on $Q$ and hence everywhere by separation. Evaluating a
retained supplied hypothesis uses only $A$. All searches at a stage are
finite: enumerate $Q$ and the current version space. Totality of the
supplied evaluators ensures termination of those searches; no uniform
runtime bound follows from totality alone.
\end{proof}

\begin{proof}[Proof of Proposition~\ref{prop:span-hardness}]
For an oracle-program index $e$, define $\phi_e(n)=1$ if $\Phi_e^A(e)$ halts within $n$ simulated steps and $0$ otherwise. This is a uniformly total $A$-computable feature map. Its range has span dimension one exactly when $e\in A'$, and dimension zero otherwise. A baseline procedure computing that dimension would decide $A'$ from $A$.

For the converse, for each of the finitely many vectors $v\in\mathbb F_2^d$, an $A$-oracle search halts precisely when some input has feature value $v$. The jump decides each such existential query. Linear algebra on the resulting finite image computes $W$. For each selected basis vector, search for a preimage using $A$; these searches terminate. All steps are uniform in the feature-program index and $d$.
\end{proof}

\begin{proof}[Proof of Proposition~\ref{prop:separator-hardness}]
Given $e$, supply the constant-zero function and
$h_e(t)=1$ exactly when $\Phi_e^A(e)$ halts within $t$ simulated steps.
Both functions are total uniformly over $A$. For any valid returned $Q$,
\[
e\in A'\quad\Longleftrightarrow\quad
\exists q\in Q\ h_e(q)=1.
\]
The reverse implication follows from bounded simulation. For the forward
implication, the functions differ somewhere if the computation halts; a
separating $Q$ must expose a difference. An effective extractor would
therefore decide $A'$ by evaluating a finite list of bounded simulations.

For the converse, for each pair of supplied programs use $A'$ to decide
whether an input on which their outputs differ exists. Such disagreement
is $A$-semidecidable under the totality promise. For every positive pair,
search relative to $A$ for one disagreement and include that input in $Q$.
There are finitely many pairs. Every inequivalent pair is distinguished,
so the resulting finite $Q$ separates the family. The construction is
uniform in the supplied indices and family size; it need not certify the
promised totality of arbitrary supplied programs.
\end{proof}

\section{Resources, tasks, and mutable agents}
\label{sec:model}

\subsection{Relative computation}

Fix a computable encoding of finite strings and tuples by natural numbers. A decision task is a set $D\subseteq\N$ with characteristic function $\chi_D:\N\to\bits$. We write $D\T A$ when an oracle Turing machine with access to the characteristic function of $A$ computes $\chi_D$ and halts on every input. Write $A\strictT B$ when $A\T B$ and $B\not\T A$. These are the standard notions of Turing reducibility and strict reducibility \citep{Turing1939,Rogers1967,Soare2016}.

For two resources define
\[
A\oplus E=\{2n:n\in A\}\cup\{2n+1:n\in E\}.
\]
An agent with access to both resources can be simulated using their join. A finite collection of fixed interfaces can likewise be encoded into a single oracle. A deterministic, total, natural-number-valued answering function can be represented by its graph: graph queries recover an answer by search, and function evaluation decides graph membership. This equivalence requires the declared totality and effective encoding.

\begin{definition}[Resource-relative layer]
The layer determined by $A$ is the principal ideal
\[
\C(A)=\{D:D\T A\}.
\]
A \emph{strict ascent} from $A$ to $B$ is the relation $A\strictT B$, equivalently $\C(A)\subsetneq\C(B)$.
\end{definition}

This layer is a class-level object. It places no bound on runtime, memory, program length, or the number of oracle queries in a terminating computation. It also does not imply that any particular agent can discover every useful program in the class. Consequently, invariance of the layer places no quantitative upper bound on practical improvement.

\subsection{Why Turing degrees enter}
\label{sec:degrees}

\begin{definition}[Turing degree, classical]
Two sets are Turing equivalent, written $A\equiv_T B$, when $A\T B$ and $B\T A$. The Turing degree of $A$ is its equivalence class
\[
\deg_T(A)=\{B:B\equiv_T A\}.
\]
Reducibility induces a partial order on these classes \citep{Rogers1967,Soare2016}.
\end{definition}

\begin{proposition}[Degrees identify equal resource-relative reach]
\label{prop:degreeclosure}
For any resources $A$ and $B$,
\[
\deg_T(A)=\deg_T(B)\quad\Longleftrightarrow\quad\C(A)=\C(B).
\]
\end{proposition}

\begin{proof}
Mutual reducibility gives equality of the layers by transitivity. Conversely, equality of the layers gives $A\in\C(B)$ and $B\in\C(A)$, since each resource is reducible to itself.
\end{proof}

Thus Turing degree is the invariant obtained by identifying resources with the same exact relative reach. It removes differences of representation that do not change which tasks are computable. Computably equivalent recodings preserve this invariant. It is neither a numerical intelligence score nor an ordering of efficiency, and two degrees can be incomparable. A degree belongs to the specified formal domain; attributing one to a physical artifact without specifying that domain leaves the mathematical object undefined.

The framework uses this ordering relationally. It can establish $A\strictT B$ from a separating task without identifying either resource with a named degree such as $\mathbf{0}'$. The degree $\mathbf{0}$ is the least degree, containing the computable sets, but the theory need not choose it as the starting point for a progression. Every result is parameterized by an arbitrary baseline $A$.

\subsection{Complete mutable state}

\begin{definition}[Effective revision system]
An effective revision system over $A$ has a finite initial configuration $s_0(x)$, obtained effectively from input $x$, and an $A$-computable partial transition map $T^A$. Its configurations are finite encodings of the complete mutable state. A run evolves by
\[
s_{t+1}(x)=T^A(s_t(x))
\]
whenever the transition is defined. Recognizing a terminal configuration and reading its output are effective operations. A run may diverge, including during an attempted revision.
\end{definition}

For an agent, a configuration may encode
\[
s=(p,h,m,d,r,v,i,\ell),
\]
where the entries describe a solver, harness, memory, stored data, trainer, evaluator, improver, and execution history. Finite parameter arrays and optimizer state are included in these descriptions. The partition is descriptive; the proof needs only a finite encoding of complete state. The inclusion of the improver follows the system-level perspective of \citet{Liu2026RSISurvey}. Our additional requirement is uniform effectivity relative to the declared resource.

For a neural learner, $p$ can contain the architecture and parameters, $r$ the training code and optimizer state, and $d$ the stored examples. Inference evaluates the encoded model; training updates its parameters; agent revision can replace the trainer or improver. New teacher replies or environmental observations enter through the declared resource access. This representation models implemented numerical operations, without assuming exact access to arbitrary real constants.

For example, a transition can interpret the current program $i_t$ on the other state components. Its output can replace $i_t$ with $i_{t+1}$, change the evaluator, or install a new interpreter. A universal machine can execute each of these changes. Keeping an effective transition semantics does not require keeping an agent's improvement algorithm fixed. It asserts that executing each encoded change is itself effective relative to $A$. An externally supplied program-selection rule that is not computable from $A$ does not satisfy that assertion.

Figure~\ref{fig:mutable-state} depicts the improver as part of the state being interpreted. The resource specification covers all answers used by the complete evolving process.

\begin{figure}[htbp]
\centering
\begin{tikzpicture}[
font=\small,>=Stealth,
state/.style={draw,rounded corners=3pt,fill=blue!5,align=center,
text width=4.7cm,minimum height=1.65cm,inner sep=7pt},
every path/.style={line width=.65pt}]
\node[draw,rounded corners=3pt,fill=green!6,align=center,
text width=8cm,inner sep=6pt] (oracle) at (0,2.4)
{Declared resource $A$: fixed answer semantics};
\draw[rounded corners=4pt,fill=gray!3] (-7.35,-1.35) rectangle (7.35,1.45);
\node[font=\small\bfseries] at (0,1.1)
{One universal simulator over $A$};
\draw[<->] (oracle.south) -- (0,1.45);
\node[state] (before) at (-4.15,-.15)
{Complete state $s_t$\\[4pt]
solver, evaluator, memory, \ldots\\
\textbf{improver $i_t$}, interpreter};
\node[state] (after) at (4.15,-.15)
{Complete state $s_{t+1}$\\[4pt]
revised components\\
\textbf{improver $i_{t+1}$}, interpreter};
\draw[->] (before.east) --
node[above=4pt] {$T^A$}
node[below=4pt,align=center,font=\scriptsize]
{execute and revise\\mutable components} (after.west);
\node[draw,rounded corners=3pt,fill=green!6,align=center,
text width=13.4cm,inner sep=7pt] (result) at (0,-2.2)
{If every input $x$ reaches a terminal output $\chi_D(x)$, then $D\T A$.};
\draw[->] (0,-1.35) -- (result.north);
\end{tikzpicture}
\caption{Mutable improvement mechanisms within uniform simulation. Each configuration encodes the complete mutable state, including the improver and interpreter. The simulator can repeat these effective transitions for an adaptively chosen, finite number of steps on each terminating input; no common bound is required. When the complete process decides $D$, Theorem~\ref{thm:simulation} gives $D\T A$. The invariant is dependence on $A$, not an unchanged agent-level improvement algorithm.}
\label{fig:mutable-state}
\end{figure}

The formal theorems concern deterministic, exact decision procedures. Computable pseudorandom generation can be included in the state. An arbitrary fixed infinite random tape is an additional oracle, and cannot silently be treated as computable. Probabilistic accuracy, interactive success, and approximate answers require their own success criteria; they are discussed as empirical quantities in Appendix~\ref{sec:empirical}.

\subsection{What the oracle represents}

The resource $A$ is specified by an interface and its answer semantics. It may encode a database, an idealized task-answering domain, or an environment together with a computable query protocol. A stateful deterministic environment can be represented by an answer map on encoded histories, provided the declared history suffices to determine its answer. An uncontrolled changing environment requires an expanded specification.

Finiteness of a program description does not imply that every semantic object associated with it is computable. A universal interpreter $U$ has finite code and effective transitions, yet its halting domain
\[
K_U=\{\pair{e}{x}:U(e,x)\downarrow\}
\]
has degree $\mathbf{0}'$. The interpreter recognizes positive halting instances; it does not decide membership in $K_U$ on every input. A finite implementation, its executed behavior, and a semantic set defined from that behavior are distinct objects.

The choice of semantic object therefore matters. The exact input-output function of a total digital inference procedure with computable operations is computable, even on an unbounded input domain. Its halting, correctness, or limiting properties are different objects and require separate analysis. Replacing a finite file by an unbounded domain alone does not establish a higher degree, while the file's finiteness alone does not exclude a higher-degree semantic property.

This permits two uses of the framework. For fully specified effective agents with computable resources, it proves closure at degree $\mathbf{0}$. For idealized answering domains, it gives conditional statements about relative reach and strict progression. Finite experiments can inform the choice of abstraction without determining the degree of its oracle.

For a staged interpretation, write $S_n\rightsquigarrow A_n$ for the association of a concrete capability stage with an idealized answering domain. This is a semantic modeling postulate, distinct from an identity between a checkpoint and a set or a claim that the checkpoint implements an exact noncomputable oracle. The postulate must specify what the domain represents and which observations support its use. Computability theory determines the consequences of the formal premises; evidence bears on whether the stage association is informative. This distinction permits a theory of idealized capability progression without treating Turing degrees as physical measurements.

An application must fix the query encoding, answer semantics, and admissible task family independently of the desired degree comparison. It must also distinguish answers an agent executes from properties defined by quantifying over its possible executions. Section~\ref{sec:acquisition} supplies an explicit example: a feature program and target family determine every answering set before any learning trajectory is observed. Its capability acquisition is exact, but its degree is preserved. A different stage association requires its own semantic justification.

\subsection{LLM services and answer contracts}
\label{sec:llm-services}

Fix a digital LLM implementation $M$, including its encoded parameters, tokenizer, effective numerical operations, and deterministic decoding with specified tie breaking. Include retained context in the query $q\in\Sigma^*$, and allow no undeclared external access. Let $b(q)\in\N$ be a total computable per-query step allowance covering preprocessing, inference, decoding, and output parsing. Bounded execution defines
\[
F_{M,b}:\Sigma^*\longrightarrow\Sigma^*\cup\{\bot\},
\qquad
G_{M,b}=\{\pair{q}{a}:F_{M,b}(q)=a\},
\]
where $\bot$ denotes timeout, invalid input, or execution failure, and the graph uses effective encodings including $\bot$. Simulating $b(q)$ steps yields the unique answer or failure, so $F_{M,b}$ is total computable. A finite number of calls alone would not guarantee this totality. Evaluation decides membership in $G_{M,b}$; conversely, enumerating encoded answers and querying the graph recovers the answer by totality.

This graph oracle records \emph{actual outputs}, including mistakes. An exact semantic oracle for a decision task $D$ instead promises $\chi_D(q)$ on every task input, with neither erroneous answers nor failure. Interpreting the model service that way requires a specified task domain, a correct encoding and parser, and an all-input correctness and completion argument. A proved implementation or exhaustive verification of a finite interface with a valid composition argument can justify such a contract. Finite benchmark success alone cannot justify it on an unrestricted domain. A computable service satisfying the contract decides a computable $D$; it cannot implement an exact noncomputable semantic oracle.

For a fallible service, specify a correctness relation $R(q,a)$ and its domain of validity. A randomized service additionally needs a randomness model and a quantified error guarantee, distinguishing per-query or adaptive-history guarantees from distributional accuracy. A sound checker can certify accepted answers without guaranteeing that a correct answer will ever be produced. These contracts do not automatically yield the exact total-decision semantics used here.

Since $G_{M,b}$ is computable, standard oracle simulation gives
\[
\C(A\oplus G_{M,b})=\C(A).
\]
Theorems~\ref{thm:simulation} and~\ref{thm:resources} apply this accounting to the complete evolving agent: generated solvers, evaluators, and improvers are encoded state, and their effective revision preserves the layer whenever the whole process decides a task. Model updates are covered when their construction and execution are uniform over $A$. Acquisition, efficiency, and reuse can improve within this layer. Query costs remain relevant to Appendix~\ref{sec:empirical}; graph-oracle equivalence is not a cost-preserving reduction. Retrieval, tools, or teacher services supplying further information $E$ must be accounted for through $A\oplus E$. Even added access does not imply strict ascent without the separating-task and preservation premises of Proposition~\ref{prop:ascent}.

\section{Uniform simulation of self-revision}
\label{sec:simulation}

\begin{theorem}[Uniform simulation]
\label{thm:simulation}
Fix an oracle $A$ and an effective revision system over $A$.
\begin{enumerate}
\item Its terminal-output function is partial $A$-computable. If it halts on every input and returns $\chi_D(x)$, then $D\T A$.
\item The same conclusion holds when a total decision procedure adaptively generates candidate solvers, evaluates them, revises its own improver, and selects a solver to run. No fixed bound on the number of revisions is required.
\item If an $A$-computable construction produces a solver index $e_t$ at each finite stage $t$, every total decision function $\Phi_{e_t}^A$ belongs to $\C(A)$.
\end{enumerate}
\end{theorem}

\begin{proof}
A single oracle Turing machine computes $s_0(x)$, simulates $T^A$ one transition at a time, and returns the output on reaching a terminal configuration. If a transition diverges or no terminal configuration is reached, the simulation diverges. Thus it computes the same partial function. When every run terminates with the correct bit, it decides $D$ relative to $A$.

For the second assertion, the simulator stores candidate descriptions, evaluation results, and the current improver as ordinary state. It executes the same effective selection and stopping rules, then simulates the selected solver using $A$. Termination on every input, rather than a uniform numerical bound on the number of revisions, suffices. For the third assertion, generate $e_t$ and apply the universal oracle interpreter. Each total resulting solver is an $A$-oracle decider.
\end{proof}

The first two assertions concern the complete adaptive process, including the cost of choosing a solver. Merely observing that each selected solver is computable is insufficient if its selection is not uniform. For any set $D$, each input $x$ has a constant program that returns $\chi_D(x)$. The existence of these individual programs does not decide $D$: selecting the right one uniformly may require the information in $D$ itself.

\begin{corollary}[Internal revision cannot decide the relative halting problem]
\label{cor:jump}
Let $A'=\{e:\Phi_e^A(e)\downarrow\}$. No total decision procedure realized by an effective revision system over $A$ decides $A'$.
\end{corollary}

\begin{proof}
The classical jump theorem gives $A\strictT A'$ \citep{Turing1939,Soare2016}; Theorem~\ref{thm:simulation} would otherwise give $A'\T A$.
\end{proof}

The result covers repeated revisions to an improver and arbitrary finite meta-depth. It does not assert that revision is unproductive, that all agents in one layer have equal performance, or that the number of useful improvements is finite. It isolates which resource-relative decision problems can be realized by the whole effective process.

\begin{example}[Expanding competence within one layer]
Let $D$ be any computable decision task. At stage $t$, an agent can build a table of $\chi_D(0),\ldots,\chi_D(t)$ and consult it on later requests, using a slower general solver elsewhere. Its cache and the input range answered quickly can grow indefinitely. Both the construction and every resulting solver remain computable. This elementary example separates improvement under a time budget from degree ascent without imposing a ceiling on the extent of useful caching.
\end{example}

\section{Accounting for external resources}
\label{sec:resources}

Self-improvement often uses tests, retrieved documents, experiments, or teacher feedback. Their origin matters: code written internally does not make an externally accessed answer internally computable.

\begin{theorem}[Resource accounting]
\label{thm:resources}
Suppose the initial state, transitions, evaluation, and final output of a revision process are uniformly computable relative to $A\oplus E$. If the complete process decides $D$, then
\[
D\T A\oplus E.
\]
If $E\T A$, then $D\T A$. The same statements hold for countably many resources $E_0,E_1,\ldots$ when access is through their effective join
\[
R=\{\pair{j}{n}:n\in E_j\},
\]
with $E$ replaced by $R$.
\end{theorem}

\begin{proof}
Replace every resource call by the corresponding query to the join and apply Theorem~\ref{thm:simulation}. If $E\T A$, an $A$-oracle machine simulates its answers, so transitivity gives $D\T A$. For the indexed family, the pair $\pair{j}{n}$ identifies the requested component and query uniformly.
\end{proof}

Uniformity in the resource index is essential. Having $E_j\T A$ for every $j$ separately does not imply $R\T A$. For example, choose a set $H\not\T A$ and let $E_j$ be $\N$ or $\varnothing$ according as $j\in H$ or not. Each $E_j$ is computable, but their join computes $H$. The missing information is the uniform choice of which constant oracle to use.

\begin{corollary}[A fixed finite artifact cannot supply strict ascent]
\label{cor:finite}
Let $\sigma$ be any fixed finite string. If a total solver has only access to $A$ and $\sigma$, its decision task lies in $\C(A)$, regardless of how $\sigma$ was obtained. In particular, finitely many training answers stored in a fixed artifact cannot, by themselves, turn such a solver into a decider for a set outside $\C(A)$.
\end{corollary}

\begin{proof}
Hardcode $\sigma$ into an oracle program and simulate the solver relative to $A$. Equivalently, the finite resource encoding $\sigma$ is computable, so its join with $A$ has the same degree as $A$.
\end{proof}

This corollary concerns a fixed artifact used uniformly on all future inputs. A family of advice strings selected separately for each input or input length may encode an uncomputable set; that selection is an additional resource. Similarly, any one interaction transcript is finite, but access to answers on an unbounded family of future queries need not be equivalent to a single finite transcript.

Consequently, a literal account of digital training with finite data and an effective deployment procedure does not establish strict degree ascent. Training can change efficiency, accuracy, generalization, and dependence on particular tools. To represent it as strict ascent between idealized task domains requires additional semantics and a nonreducibility premise. That premise cannot be inferred merely from a change to weights.

\begin{example}[Changed access with unchanged code]
An agent that simply forwards membership queries to $E$ decides $E$. If $E\not\T A$, adding that interface changes its available reach even though its controller is unchanged. By comparison, replacing a controller by a much better $A$-computable program can improve performance without adding any oracle information. Code change and resource change are therefore separate coordinates.
\end{example}

Resource accounting is relative to the declared boundary. If a teacher or environment is included in the baseline, the whole process is analyzed over that enlarged resource. If it is excluded, its answers must be recorded as additional access. The distinction depends on the resource specification, rather than on whether the supplier is human, another model, or an experimental instrument.

\subsection{Construction, acquisition, and deployment assistance}
\label{sec:assistance}

Assistance can enter when a capability is constructed, when its parameters or
procedures are acquired, or when it is deployed. Construction includes choosing
a task specification, designing a learner, and supplying a completion
certificate. Acquisition includes teacher answers and experimental feedback.
Deployment includes queries made while solving subsequent tasks. Removing a
teacher after acquisition establishes independence from that deployment
interface; it does not establish independence from the information or design
used to construct the solver. Conversely, historical human authorship of a
finite program does not establish that a human is necessary for every later
execution or improvement.

The resource bound includes assistance at all three stages. If construction,
acquisition, and deployment are uniformly effective over
$A\oplus E_{\rm c}\oplus E_{\rm a}\oplus E_{\rm d}$, their composite decision
procedure is reducible to that join by Theorem~\ref{thm:resources}. A fixed finite
construction artifact can be hardcoded by Corollary~\ref{cor:finite}; uniform
selection of artifacts across tasks remains part of the account. If the
composite decides $D\not\T A$, the combined assistance cannot be reducible to
$A$. This necessary condition identifies additional computational information,
without identifying its supplier as human or showing that humans can supply it.

\begin{proposition}[Information required for exact task completion]
\label{prop:assistance-information}
Let $W$ be a set of possible task situations, let $v:W\to O$ be the complete
observation available to a fixed decision rule, and let $\ell:W\to\bits$ be the
required answer. There exists a function $f:O\to\bits$ satisfying
$f(v(w))=\ell(w)$ for every $w$ if and only if $\ell$ is constant on every
fiber of $v$. In particular, if $v(w_0)=v(w_1)$ and
$\ell(w_0)\ne\ell(w_1)$, no such rule exists. For any additional observation
$h:W\to H$, an exact rule $g(v(w),h(w))=\ell(w)$ requires
$h(w_0)\ne h(w_1)$ for every such pair.
\end{proposition}

\begin{proof}
An exact rule gives
$\ell(w_0)=f(v(w_0))=f(v(w_1))=\ell(w_1)$ on each fiber.
Conversely, assign each observation in the range of $v$ the common label of
its fiber, and assign an arbitrary bit elsewhere. For the final assertion,
equal original and additional observations would give equal pairs and hence
equal answers, contradicting the differing required labels.
\end{proof}

The sufficiency direction is a statement about mathematical functions; it
supplies neither an effective decoder nor an efficient method for discovering
one. The impossibility direction applies even without a computability bound.
The observation $v$ must include available construction artifacts, stored
reasoning, and the intervention history. Any representation or synthetic
reasoning data computed solely from $v$ remain identical on an
indistinguishable pair. Such computation can expose useful consequences and
improve performance within a budget, but cannot guarantee an absent task bit.
The proposition concerns exact answers; distributional or approximate
guarantees require their own loss and sampling assumptions.

\paragraph{When human input is necessary.}
Suppose a human privately chooses a goal bit $b$, every channel available
without that person's input is identical for $b=0$ and $b=1$, and correctness
requires returning $b$. Proposition~\ref{prop:assistance-information} then
excludes a guaranteed solution without information from that human. A single
truthful message carrying $b$ suffices: install the constant rule returning
$b$ for that fixed goal. This example establishes necessary human-origin
information under explicit source assumptions, while permitting all later
executions to proceed without further intervention. If another channel
reveals $b$, the human-source necessity premise no longer holds. In applications,
exclusive access to a goal or authorization is a property of the declared
workflow, not a consequence of halting undecidability.

The jump obstruction and this information obstruction address different
claims. The former excludes a total baseline-effective solver even when the
task is completely specified; the latter excludes recovery of a distinguishing
answer absent from the available observation. Neither proves that humans
decide arbitrary halting instances or that every useful improvement requires
human involvement. Difficulty formalizing a capability-construction mechanism
does not establish that it supplies non-baseline information: its effectivity,
correctness, and resource contract remain separate obligations.

\section{Limiting revision and recognizable completion}
\label{sec:limits}

\begin{definition}[Limiting revision]
A limiting revision over $A$ is a total $A$-computable function
$g:\N^2\to\bits$ such that, for each $x$, the sequence $g(x,t)$ is eventually constant. Its outcome is the set $D$ defined by
\[
\chi_D(x)=\lim_{t\to\infty}g(x,t).
\]
\end{definition}

Totality here applies to each finite approximation $g(x,t)$. It does not assert a computable bound on the stage of its last change. This is the classical setting of trial-and-error and limit computation \citep{Shoenfield1959,Putnam1965,Gold1967,Soare2016}.

\begin{theorem}[Relativized limit lemma, classical]
\label{thm:limit}
A set $D$ is the outcome of a limiting revision over $A$ if and only if $D\T A'$.
\end{theorem}

The theorem characterizes a class of outcomes. Some outcomes are already in $\C(A)$: the constant approximation $g(x,t)=\chi_D(x)$ suffices whenever $D\T A$. Others are not. For example,
\[
g(e,t)=
\begin{cases}
1,&\text{if }\Phi_e^A(e)\text{ halts within }t\text{ simulated steps},\\
0,&\text{otherwise}
\end{cases}
\]
is total and $A$-computable, and converges to $\chi_{A'}$. Each input changes at most once, but a current zero does not certify that it will remain zero.

\begin{proposition}[Effective completion collapses to the base layer]
\label{prop:modulus}
Suppose $g$ is a limiting revision over $A$ with outcome $D$.
\begin{enumerate}
\item If an $A$-computable function $m(x)$ satisfies
\[
t\geq m(x)\ \Longrightarrow\ g(x,t)=\chi_D(x),
\]
then $D\T A$.
\item More generally, if an $A$-computable stopping procedure halts on every $x$ and returns a stage $\tau(x)$ with $g(x,\tau(x))=\chi_D(x)$, then $D\T A$.
\end{enumerate}
\end{proposition}

\begin{proof}
For the first assertion compute $g(x,m(x))$. For the second run the stopping procedure and then compute $g(x,\tau(x))$. Each construction is a total $A$-oracle decision procedure for $D$.
\end{proof}

The second assertion does not require a certificate that all future guesses agree. Correct selection of a final answer already suffices. When $D\not\T A$, neither kind of uniformly effective completion is available for that approximation.

\begin{proposition}[No universal effective stabilization certifier]
\label{prop:certifier}
There is no $A$-computable procedure which, given an index for any total pointwise-convergent binary $A$-computable approximation and an input $x$, always terminates with a stage after which the approximation on $x$ never changes. This remains impossible when the approximation changes at most once on each input.
\end{proposition}

\begin{proof}
Apply such a procedure to the fixed index for the halting approximation above. It would yield an $A$-computable stabilization bound $m(e)$, and evaluating $g(e,m(e))$ would decide $A'$ relative to $A$, contradicting the jump theorem. The approximation used in this reduction changes at most once per input.
\end{proof}

Thus a convergence guarantee and a stopping guarantee have different computational content. An empirical rule such as stopping after several unchanged scores is a finite effective rule and remains covered by Theorem~\ref{thm:simulation}. It may be reliable for a restricted family under additional assumptions; the impossibility result concerns a procedure correct on the entire stated class.

Finite recursive depth also differs from nested limiting semantics. Composing ten, a thousand, or an adaptively chosen finite number of effective meta-levels stays within the simulation theorem. Repeated \emph{limit operators} define a different hierarchy: fixed $k$-fold nested limits characterize $\C(A^{(k)})$ when the approximation is uniformly $A$-computable and all specified iterated pointwise limits exist, as in Appendix~\ref{app:nested}. No uniform convergence rate is assumed. A layer in an agent architecture is not automatically a limit operator.

\begin{figure}[htbp]
\centering
\begin{tikzpicture}[
font=\small,>=Stealth,
regime/.style={draw,rounded corners=3pt,align=center,text width=4.6cm,
minimum height=2cm,inner sep=7pt},
every path/.style={line width=.65pt}]
\node[regime,fill=blue!5] (effective) at (0,0)
{Effective self-revision\\[3pt]mutable solver and improver\\decided task: $D\T A$};
\node[regime,fill=green!6] (resource) at (6.3,0)
{Additional resource $E$\\[3pt]effective access to $A\oplus E$\\decided task: $D\T A\oplus E$};
\node[regime,fill=orange!8] (limit) at (0,-3.05)
{Pointwise limiting revision\\[3pt]$\chi_D(x)=\lim_t g(x,t)$\\outcome: $D\T A'$};
\node[regime,fill=gray!6] (ascent) at (6.3,-3.05)
{Strict ascent witness\\[3pt]$D\not\T A,\quad D\T B$\\and $A\T B$ imply $A\strictT B$};
\draw[->] (effective.east) -- node[above,align=center,font=\scriptsize] {new\\access} (resource.west);
\draw[->] (effective.south) -- node[left,align=center,font=\scriptsize] {change of\\output semantics} (limit.north);
\draw[->,dashed] (resource.south) -- node[right,align=center,font=\scriptsize] {requires witness\\and preservation} (ascent.north);
\end{tikzpicture}
\caption{Resource access and output semantics determine the formal conclusion. Neither adding a resource nor taking a limit alone guarantees strict ascent. The lower-right implication requires an independent task witness and preservation of the predecessor resource; no general identification of the successor with $A'$ follows.}
\label{fig:regimes}
\end{figure}

\section{Internalization and strict relative ascent}
\label{sec:ascent}

A resource-dependent capability can become available through a successor interface. The mathematical content of this transition depends on which task is retained and which predecessor capabilities are preserved.

\begin{definition}[Essential support and internalization witness]
A task $D$ has \emph{essential support} $E$ relative to $A$ when
\[
D\T A\oplus E\qquad\text{and}\qquad D\not\T A.
\]
A proposed successor resource $B$ \emph{internalizes this task} when $D\T B$ without separate access to $E$. The \emph{preservation condition} is $A\T B$.
\end{definition}

Internalization here means availability through the successor's declared interface. It does not specify a physical mechanism such as distillation, nor imply that the whole resource $E$ has been absorbed.

\begin{proposition}[Witness criterion for strict ascent]
\label{prop:ascent}
If $D\not\T A$, $D\T B$, and $A\T B$, then
\[
A\strictT B,\qquad \C(A)\subsetneq\C(B).
\]
If also $B\T A\oplus E$, then $A\strictT B\T A\oplus E$.
\end{proposition}

\begin{proof}
If $B\T A$, then $D\T B\T A$, contradicting $D\not\T A$. Combining $B\not\T A$ with preservation gives strictness. Transitivity gives inclusion of the ideals, and $B$ belongs only to the larger one. The additional upper bound is exactly the additional construction premise.
\end{proof}

The implication is an elementary consequence of reducibility. Its application requires separate justification of each premise. Failure of one particular predecessor agent to solve $D$ does not prove $D\not\T A$, since the latter quantifies over every effective procedure relative to $A$. Preservation is also substantive: a successor can acquire one task while losing another. Without $A\T B$, the two resources can be incomparable.

\begin{example}[An exact idealized transition]
Let $E\not\T A$ and put $B=A\oplus E$. Taking $D=E$ yields essential support, preservation, and internalization through the joined interface, so $A\strictT B$. This constructs an exact resource transition. Its stronger reach comes from retaining access to $E$; it is not a construction of noncomputable capability from a fixed finite training artifact.
\end{example}

The jump has two precise roles: $A'$ is a canonical task outside $\C(A)$, and it bounds the outcomes of one limiting revision over $A$. It is not a universal bound on a successor. For example, the criterion permits $B=A''$ and hence $B\not\T A'$. More generally, even if a particular internalized witness satisfies $D\T A'$, the successor may contain other information. That bound on $D$ does not give $B\T A'$. Equality with the jump likewise requires an additional argument.

A sequence of transitions satisfying the witness criterion gives
\[
A_0\strictT A_1\strictT A_2\strictT\cdots.
\]
Each transition is conditional on its own premises. The local result supplies neither a computable method for producing this sequence from $A_0$ nor a guarantee that empirical development follows a monotone chain. A uniform effective construction with only access to $A_0$ remains subject to Theorem~\ref{thm:simulation}.

\section{Proposal services and acquisition context}
\label{app:proposal-context}

\subsection{Neural computation and procedural models}

Generative models have two roles in this organization: they construct candidate procedures, and they supply answers that those procedures consult. A generated program becomes an effective procedure through its interpreter; a model answer has the semantics of the implemented service until a stronger correctness contract is justified. \citet{LaurentPlatzer2025Oracular} already combine procedural strategies, LLM-guided search policies, and demonstrations in \emph{Oracular Programming}. The unpublished companion manuscript \emph{The Return to Procedure} \citep{Lu2026Return} develops the historical context and architectural intuition for this dual role. Here we study its resource dependence when the agent's procedures themselves evolve.

Neural universality results address a different direction of simulation. \citet{Perez2021Attention} construct hard-attention Transformers with specified positional encodings, arbitrary-precision rational activations, and unbounded execution. \citet{Schuurmans2023Memory} uses unbounded external associative read--write memory and exact instruction responses. \citet{Schuurmans2024Autoregressive} and \citet{Lewandowski2026Universal} use extended autoregressive decoding: a bounded window reads successive prefixes of a growing operational queue, appending responses that can contain multiple symbols per step. Their universal simulations require exact encoded transition rules, unbounded accessible storage, and no fixed overall runtime cap; retaining only a fixed recent tail does not implement that queue. In \citet{Lewandowski2026Universal}, the simulation theorem assumes a suitable injective encoding and exact rule responses, with learned codebooks used for the randomly initialized models. These constructions do not establish universality for every fixed deployed model or correctness on arbitrary semantic tasks. Our premise is instead that an oracle machine can execute the specified agent process; no individual neural component needs to be universal.

\paragraph{LLM proposals and the exact teacher.}
An LLM service from Appendix~\ref{sec:llm-services} can propose queries, rewrite the proposer, or suggest a coefficient vector for a solver. To retain the theorem's information model, its initial parameters and other resources are target independent; any target-specific information arrives through the recorded teacher replies. The bounded guard tolerates malformed, redundant, or unsuccessful proposals by using its supplied-basis fallback. For a proposed solver, deployment uses the prescribed feature evaluator and checks $\widehat\theta\in\Theta_T$ and $V_T=W$. Passing finite tests does not certify arbitrary generated code.

Figure~\ref{fig:llm-acquisition} separates these roles. The teacher's contract is the exact bit $f_\theta(x)$ for the fixed target. If an LLM supplies labels too, those labels need an independent justification of that contract; the proposal role supplies none. Rank alone cannot detect incorrect labels or certify the true target after noisy replies. Thus the same architecture can use learned generation for acquisition while deriving teacher removal from explicit structural and answer assumptions. The experiments in Section~\ref{sec:experiments} use symbolic proposers and exact synthetic teachers; they do not evaluate an LLM implementation.

\begin{figure}[htbp]
\centering
\begin{tikzpicture}[
font=\small,>=Stealth,
role/.style={draw,rounded corners=3pt,align=center,text width=3.4cm,
minimum height=1.2cm,inner sep=6pt},
every path/.style={line width=.65pt}]
\node[role,fill=blue!5] (llm) at (-5.2,0)
{LLM service\\candidate queries,\\proposers, or solvers};
\node[role,fill=gray!6] (guard) at (0,0)
{Protected protocol\\bounded checks,\\basis fallback};
\node[role,fill=orange!8] (teacher) at (5.2,0)
{Exact teacher\\$x\mapsto f_\theta(x)$\\declared contract};
\draw[->] (llm.east) -- node[above,font=\scriptsize] {propose} (guard.west);
\draw[->] (guard.east) -- node[above,font=\scriptsize] {$x$} (teacher.west);
\node[draw,rounded corners=3pt,fill=green!6,align=center,
text width=8.4cm,inner sep=7pt] (cert) at (0,-2.35)
{Stored constraints $T$ and known span $W$\\[3pt]
$\widehat\theta\in\Theta_T,\ V_T=W$
$\ \Longrightarrow\ $ correct teacher-free solver};
\draw[->] (guard.south) -- node[right,font=\scriptsize] {checked candidate} (cert.north);
\draw[->] (teacher.south) |- node[pos=.25,right,font=\scriptsize] {exact label} (cert.east);
\end{tikzpicture}
\caption{Two roles in certified acquisition. Learned proposals affect the acquisition procedure; exact teacher replies supply target constraints. Theorem~\ref{thm:query-span} certifies teacher removal using those constraints and the independently justified feature span. A proposal service is not thereby an exact teacher.}
\label{fig:llm-acquisition}
\end{figure}

\subsection{Relation to learning and internalization}

The query-span criterion and its counting lower bound are elementary instances of exact query learning \citep{Angluin1988Queries}. Their role here is to make teacher removal and improver mutability explicit within one resource account. Learning-to-learn models analyze how experience across tasks can improve inductive bias \citep{Baxter2000Bias}; the present protocol instead holds the feature language fixed and permits its query policy to change. Generalized distillation studies transfer from teachers and privileged information \citep{LopezPaz2015Distillation}; our example uses exact realizability and a known span to certify agreement on every input. These assumptions are stronger than finite test accuracy and are what permit the completion conclusion.

In Appendix~\ref{sec:ascent}, availability through a successor interface alone does not specify an acquisition mechanism. Here the mechanism is explicit: the constraint transcript determines a finite coefficient vector, and the deployment procedure removes the teacher interface. This is retained task capability within $\C(A)$, rather than the strict-ascent case of Proposition~\ref{prop:ascent}. The restriction is informative: exact acquisition, reliable stopping, and improved query selection need not coincide with an increase in Turing degree.

\section{Reading contemporary self-improvement systems}
\label{sec:cases}

These applications identify which components change and which assumptions would license a formal classification. They do not assign measured Turing degrees to deployed systems. The recent preprints provide architectural examples, while the results above are independent of their reported benchmark scores.

\subsection{Autonomy, effectiveness, and computational reach}
\label{sec:roadmap}

\citet{Duan2026LastAI} organize recursive self-improvement around increasing autonomy, culminating in revision of the improvement mechanism. Their distinction between \emph{structural L5} and \emph{effective L5} separates inherited, reused mechanism changes from demonstrated gains in successor quality under comparable budgets and independent assessment. They identify reliable accumulation across generations as an open empirical question. This operational account is compatible with the present theory: neither classification asserts a separation between resource-relative computability classes.

Table~\ref{tab:rsi-properties} separates the corresponding claims. Autonomy concerns who selects and executes revisions; effectiveness concerns what those revisions achieve; strict ascent concerns which exact decision tasks become available. These properties require different evidence. In particular, a successful experiment on improvement of the improver does not establish that the resulting task lies outside $\C(A)$.

\begin{table}[htbp]
\centering
\small
\caption{Three properties of self-improvement and their distinct evidential requirements. Here $A$ is the predecessor resource, $B$ the successor resource, and $D$ a separating decision task.}
\label{tab:rsi-properties}
\begin{tabularx}{\linewidth}{@{}p{0.22\linewidth}XX@{}}
\toprule
Property & Evidence or formal condition & Relation to computational reach \\
\midrule
Autonomy over improvement & Persistent, agent-selected changes to a mechanism that governs later revisions. & Mutable mechanisms remain covered by uniform simulation when their execution is effective over $A$. \\
\addlinespace
Practical effectiveness & Better successor performance or lower cost under controlled starting states, resource access, and assessment. & Gains can occur within a fixed $\C(A)$; their rate requires a resource-bounded analysis. \\
\addlinespace
Strict relative ascent & $D\not\T A$, $D\T B$, and $A\T B$, justified for specified formal domains. & The witness and preservation conditions imply $\C(A)\subsetneq\C(B)$. \\
\bottomrule
\end{tabularx}
\end{table}

Even an agent that autonomously designs its successors and repeatedly improves the procedure used to design them can remain within one resource-relative layer. Theorem~\ref{thm:simulation} applies to the complete process, including changing interpreters and evaluators, rather than to a fixed agent-level optimizer. Conversely, an unchanged program can gain access to stronger answering resources. Thus the extent of software autonomy does not determine the degree of the resource base. Closure at a fixed degree leaves open how much useful AI development can be automated and how rapidly practical capability can grow.

\subsection{Representative architectures}

\paragraph{G\"odel machines and mutable improvement rules.}
The G\"odel-machine proposal permits modifications to the system's own software, including its proof-search machinery, when a suitable benefit statement is proved in its formal setup \citep{Schmidhuber2003Godel}. In the present model, both the proposed program and the proof searcher are state. If the proof search, execution, and environmental access are effective relative to $A$, the whole process falls under Theorem~\ref{thm:simulation}. A proof of improvement relative to a utility specification is distinct from a proof that an agent can decide a task outside its available resource-relative class.

\paragraph{Darwin G\"odel Machine.}
DGM searches over self-modified coding agents, retains an archive, and uses benchmark evaluation in place of requiring a proof of utility improvement. Its reported implementation uses frozen pretrained foundation models \citep{Zhang2025DGM}. Agent code and archived variants fit the mutable state; model calls, repositories, and evaluation feedback belong in the resource specification. Conditional on effective access to a fixed such specification, archive growth and improver edits remain covered by uniform simulation. A frozen model alone is insufficient to establish a fixed resource base if the environment supplies additional information.

\paragraph{Meta$^n$.}
Meta$^n$ repeatedly applies a fixed meta-operation to generated solver layers and their traces, choosing depth through its operational convergence criterion \citep{Kim2026Metan}. Its architectural recursion illustrates why increasing meta-depth can be useful while remaining uniformly simulable. A terminating depth-selection rule is covered by Theorem~\ref{thm:simulation}. It does not instantiate the nested limit operators in Appendix~\ref{app:nested} merely because its layers are described as meta-levels. Any claimed correspondence to those operators would require the stated pointwise limit semantics.

\paragraph{Dream-RSI.}
Dream-RSI makes exploration policies explicit, improves them through replay simulations built from discovery histories, and redeploys them while leaving the underlying coding agent unchanged \citep{Zheng2026DreamRSI}. The policy and stored discovery tree are mutable state. Replaying a stored finite history adds no independent oracle information; online evaluations require an account of their source. The framework consequently distinguishes improvement in the use of accumulated experience from changes in the answers available through the environment.

\paragraph{AlphaEvolve and improvements to training infrastructure.}
AlphaEvolve combines model-generated code, automated evaluation, and evolutionary search. Its reported applications include improvements to algorithms and infrastructure used in AI training \citep{DeepMind2025AlphaEvolve}. Such feedback can improve the cost of producing later models. In our framework, a more efficient trainer remains an effective procedure over its declared resources. Whether a later model has internalized a particular externally supported capability is a separate comparison; a training-cost improvement alone does not establish the ascent premises.

These examples connect the modification scope to the resource and output specification. The former identifies which components change; the latter determines the applicable computability result. Architectural descriptions therefore support a conditional application of the theory, while an ascent claim additionally needs the separating task and preservation premises of Proposition~\ref{prop:ascent}.

\section{Practical capability and the limits of empirical attribution}
\label{sec:empirical}

\subsection{Performance within a layer}

Let
\[
P(s,b,\mathcal D,E)
\]
denote expected task performance of a specified agent state $s$, under deployment budget $b$, task distribution $\mathcal D$, and declared external support $E$. The expectation is over the sampled task and the declared randomness of the agent and environment; the score lies in $[0,1]$. The budget includes execution and support-access costs, with their units specified by the evaluation. Here $E$ names an experimental intervention, such as enabling a teacher interface; relating it to a formal oracle requires an answer specification. This is an empirical or resource-bounded quantity, rather than a Turing degree. Revisions can raise $P$ or reduce the budget required to meet a fixed target without changing the computability class.

Improvement of an improver can be measured through the performance of its descendants under matched improvement budgets. Comparing two improvers requires common starting states, resource access, and evaluation conditions; a stronger final agent alone does not identify which part of the improvement pipeline changed. The improvement budget separately includes acquisition, search, and validation costs, including the discovery and assessment of the revised mechanism. Assessment on fresh tasks can test whether its benefit transfers, while recording regressions distinguishes accumulated capability from displaced capability. These comparisons concern performance and causal attribution. The unbounded simulation theorem does not supply their rate, sample complexity, or effect size; Theorem~\ref{thm:query-span} supplies a query bound only for its stated restricted family.

This distinction preserves the importance of learning and self-improvement within a class. A proof of closure cannot show that future agentic gains will be small, that useful autonomous research is impossible, or that every practical improvement requires additional human guidance. It also cannot rank a model update above a harness update by engineering value.

\subsection{Behavioral internalization}

For a predecessor and successor, compare performance with and without the same named support $E$, holding evaluation conditions and budgets explicit. Define support lift by
\[
L_t(E)=P(s_t,b,\mathcal D,E)-P(s_t,b,\mathcal D,\varnothing).
\]
Higher unsupported successor performance, smaller support lift, and preserved supported performance form an operational pattern consistent with internalization. Each quantity requires uncertainty estimates, and the interpretation can be affected by ceiling effects, distribution changes, hidden support, or altered support use. Even the complete pattern does not uniquely identify the learning mechanism.

The comparator $\varnothing$ removes the named support, not every resource available to the agent. Persistent external memory or a generated skill library can retain a capability at the agent level without placing it in model weights. The retained artifact and the deployment boundary therefore need to be specified before attributing the gain to a model or to a whole agent.

\begin{example}[Shrinking support lift without improved use of support]
\label{ex:headroom}
Suppose an unsupported solver has accuracy $u$, and an unchanged fallback intervention corrects a fixed fraction $q$ of its errors without spoiling correct answers. Under a common deployment budget large enough to permit the fallback,
\[
P(\varnothing)=u,\qquad P(E)=u+q(1-u),\qquad L(E)=q(1-u).
\]
Let $q=1/2$ and let an independent improvement raise $u$ from $0.4$ to $0.8$. Unsupported accuracy rises, supported accuracy rises from $0.7$ to $0.9$, and support lift falls from $0.3$ to $0.1$. All three behavioral conditions hold, although the fallback's conditional effectiveness remains unchanged and the model includes no transfer from the support into the solver. Thus the pattern alone cannot identify such transfer.
\end{example}

For this rescue model, $L(E)/(1-u)$ identifies conditional rescue probability when $u<1$. This normalization is not a general internalization estimator: residual task composition, harmful interventions, and changes to support use require separate controls. In contrast, Theorem~\ref{thm:query-span} establishes teacher-free correctness from a stated family, noiseless observations, and a proved span specification. It supplies a sufficient condition for exact retention, rather than deriving a mechanism from aggregate support lift.

\begin{proposition}[Finite observations do not identify an oracle degree]
\label{prop:finiteobservations}
Fix an oracle $A$. For every finite consistent collection $F$ of membership-query answers, there exist sets $C$ and $H$ agreeing with all answers in $F$ such that
\[
C\T A,\qquad H\equiv_T A'.
\]
Thus those observations alone do not establish whether the answering set lies outside $\C(A)$.
\end{proposition}

\begin{proof}
Let $C$ contain exactly the finitely many inputs assigned answer $1$ by $F$; then $C$ is computable. Obtain $H$ from $A'$ by overwriting its values on the finite domain of $F$ with the specified answers. Finite modifications preserve Turing degree, since each finite correction can be hardcoded in either reduction. Therefore $H\equiv_T A'$. Both sets agree with $F$.
\end{proof}

The conclusion also applies to a finite transcript of adaptive membership queries: with the same answers and fixed internal randomness, the evaluator follows the same observed path under either completion. This is a statement about inference from unrestricted finite oracle observations. Structural restrictions on the candidate resources or mathematical proofs about their definitions can provide additional information.

The finite-dimensional example illustrates precisely such a restriction. Once realizability in $\mathcal H_\phi$ and the span specification are established, a spanning transcript identifies the answer on every input. The arbitrary higher-degree completions in Proposition~\ref{prop:finiteobservations} are then excluded by the model assumptions. The finite observations do not prove those assumptions; they identify a target conditional on them.

A failed benchmark instance establishes a failure of the tested procedure. The premise $D\not\T A$ says that \emph{every} effective $A$-oracle procedure fails to decide the entire task. Finite support-removal studies cannot bridge that quantifier difference unaided. The appropriate connection is therefore explicit: experiments measure resource dependence and retained performance; proofs establish consequences of formal resource and task assumptions.

\section{Acquisition evaluation details}
\label{app:acquisition-evaluation}

\subsection{Tasks and evaluation}

The feature image is $\mathbb F_2^d$, the coordinate basis is supplied, and each target coefficient is sampled uniformly. The input distribution is
\[
\mu_{H,\rho}=(1-\rho)\operatorname{Unif}(H)
                 +\rho\operatorname{Unif}(\mathbb F_2^d),
\]
where $H$ is a randomly generated subspace. The learner sees samples rather than $H$ or $\rho$. For an acquired span $V$, certified coverage is evaluated exactly as
\begin{equation}
c(V)=(1-\rho)2^{\dim(V\cap H)-\dim H}+\rho2^{\dim V-d}.
\label{eq:empirical-coverage}
\end{equation}
The true risk of a fitted coefficient is also calculated exactly from its difference from the target: a nonzero linear functional has error $1/2$ on a uniform subspace unless it vanishes there. No evaluation samples are supplied to the policies. Enumeration on 240 small-domain cases independently checks these calculations, and exhaustive candidate checks validate the greedy implementation in dimensions two through six. Appendix~\ref{app:experiment-protocol} gives the protocol, cost accounting, and uncertainty calculation.

\subsection{Stopping configurations and complete results}
We compare six rules at $d\in\{8,12,16\}$, $\dim H=d-2$, and $\rho\in\{1,.1,.01\}$, with 200 independent tasks per condition. Fixed-budget rules use $d$ or $2d$ i.i.d. acquisition labels. A prediction-stability heuristic stops after two consecutive unchanged predictions on an independent pool of 128 unlabeled inputs, checked every $d/2$ acquisition labels. A statistical rule checks a frozen solver on 252 fresh labeled inputs at each of at most eight looks; zero observed errors suffice to certify risk at most $\epsilon=.02$ with family error probability at most $\delta=.05$. Validation labels are charged separately and are not used to fit the solver. Rank-on-stream waits for the acquisition span to have rank $d$. The supplied-basis rule actively queries the known coordinate basis. Rules without a successful stopping criterion retain the teacher at the cap of $4d$ acquisition labels.

\begin{table}[!htb]
\centering
\small
\caption{Teacher removal averaged equally over $d=8,12,16$, with 600 tasks per row. ``Bad'' is the percentage of all trials that remove the teacher with true risk above 2\%. ``Risk'' is mean error conditional on removal. Labels include acquisition and validation costs, also in trials that retain the teacher. A zero observed bad-removal rate is an empirical outcome; only the stated structural or statistical assumptions supply a guarantee.}
\label{tab:stopping}
\begin{tabular}{@{}clrrrr@{}}
\toprule
$\rho$ & Rule & Remove (\%) & Bad (\%) & Risk (\%) & Labels \\
\midrule
1 & Fixed $d$ & 100.00 & 37.33 & 18.667 & 12.00 \\
 & Fixed $2d$ & 100.00 & 0.17 & 0.083 & 24.00 \\
 & Prediction stability & 100.00 & 0.67 & 0.333 & 26.09 \\
 & Statistical risk & 100.00 & 0.00 & 0.000 & 609.74 \\
 & Rank on stream & 100.00 & 0.00 & 0.000 & 13.60 \\
 & Supplied basis & 100.00 & 0.00 & 0.000 & 12.00 \\
\midrule
.1 & Fixed $d$ & 100.00 & 65.33 & 12.492 & 12.00 \\
 & Fixed $2d$ & 100.00 & 33.83 & 1.692 & 24.00 \\
 & Prediction stability & 99.83 & 32.17 & 1.611 & 27.04 \\
 & Statistical risk & 88.00 & 0.00 & 0.000 & 1033.16 \\
 & Rank on stream & 76.00 & 0.00 & 0.000 & 29.48 \\
 & Supplied basis & 100.00 & 0.00 & 0.000 & 12.00 \\
\midrule
.01 & Fixed $d$ & 100.00 & 12.83 & 6.728 & 12.00 \\
 & Fixed $2d$ & 100.00 & 0.00 & 0.350 & 24.00 \\
 & Prediction stability & 100.00 & 0.67 & 0.682 & 24.31 \\
 & Statistical risk & 90.67 & 0.00 & 0.335 & 953.17 \\
 & Rank on stream & 4.83 & 0.00 & 0.000 & 47.15 \\
 & Supplied basis & 100.00 & 0.00 & 0.000 & 12.00 \\
\bottomrule
\end{tabular}
\end{table}
\subsection{Query-program mutation and comparisons}
We next use $d\in\{8,12\}$, $\dim H=d/2$, and $\rho=.1$. Each of 40 independent replicates per dimension contains six episodes with fresh targets. The frequent subspace changes after episode three, without announcing the change to the policies. A development set contains 256 unlabeled inputs; each episode supplies 64 fresh unlabeled inputs shared by the policies. A program is an ordered list of proposed query vectors. Its protected rank guard replaces redundant proposals with an independent supplied basis vector, so all programs make exactly $d$ teacher calls and finish with zero risk.

The revised program starts from a greedy program learned on the development set and persists across episodes. At each episode it tries 64 mutations, swapping two positions or replacing one vector with a sampled input, and accepts strictly greater empirical mean coverage over budgets $0,\ldots,d-1$. The comparator ``search from basis'' uses the same mutations and online budget but restarts from the coordinate basis. Other comparators are the fixed coordinate basis, a random full-rank program, the frozen development program, and greedy selection from the current sample. Greedy repeatedly chooses the residual coset containing the most sampled inputs. Revision modifies executable query-program contents; its mutation rule and the protected guard remain fixed. Teacher replies do not guide selection because coverage in this family is target independent.

\section{Machine-checked operational and acquisition results}
\label{app:formalization}

The accompanying Lean 4~\citep{deMouraUllrich2021Lean} project uses
version 4.34.1 and Mathlib pinned to commit
\texttt{d13f23b723b8a846827a245b89c10fc7d3f11612}.
Compilation, axiom auditing, and exact manuscript coverage are separate
checks. No proof is admitted and no custom axiom is declared. The permitted
foundations are propositional extensionality, classical choice, and quotient
soundness. Complete formalization of all paper statements remains open.

\paragraph{Relative computation and revision.}
\texttt{RelativeComputability.lean} identifies a finite relative-program syntax
with Mathlib's relative-recursive closure and proves substitution, layer
identities, a diagonal nonreducibility result, and a strict jump in that model.
\texttt{PartialRevision.lean} and \texttt{ResourceAccounting.lean} check partial
state transitions, unbounded terminal-output simulation, binary and countable
joins, uniform resource access, and fixed finite-artifact absorption.
\texttt{FiniteObservations.lean} checks finite-observation completions and
preservation of degree under finite changes.
\texttt{EffectiveOracle.lean} supplies one uniform encoded evaluator for
partial computations over a total answering resource, with effective finite
oracle prefixes and unbounded search. \texttt{UniformJump.lean} proves domain
decidability relative to that evaluator's jump and its diagonal
noncomputability. \texttt{StabilizationObstruction.lean} checks a total,
monotone binary approximation with at most one change and excludes a
baseline-effective stabilization bound for it. Finite-oracle approximation
and parametric-evaluation lemmas adapt Apache-licensed OpenAI sources
\citep{OpenAIMath2026}. The standard oracle-machine representation,
acceptable-index correspondence, equivalence of the jump presentations,
and the relativized limit lemma remain open formalization obligations.

\paragraph{Acquisition.}
\texttt{QuerySpan.lean}, \texttt{RankProgress.lean}, and
\texttt{QueryDecisionTree.lean} check answer identification by span, rank
progress, and the binary-query lower bound. The executable
\texttt{CertifiedAcquisition.lean} learner retains its finite-family
termination, correctness, and $N-1$ teacher-query bound. It also checks the
finite-list reduction from separating specifications to a semidecision
existence test. The complete effective rank-guarded implementation,
uniform relative-halting instantiation, and $A'$ upper-bound extractor remain
open obligations. \texttt{LegacyOraclePrograms.lean} preserves the earlier
finite operational semantics and oracle-substitution development.

\paragraph{Information and assistance.}
\texttt{AssistanceBoundary.lean} checks
Proposition~\ref{prop:assistance-information}: exact semantic sufficiency is
equivalent to constancy on observation fibers, and successful advice must
distinguish differently labeled situations with the same original observation.
Its sufficiency proof uses classical choice and asserts no effective decoder.
Additional lemmas check the missing-goal-bit obstruction, its preservation
under internal postprocessing, and sufficiency of supplying the goal once.
Relative-recursive lemmas separately show that strict reach requires
non-baseline assistance and that baseline-computable assistance cannot solve
the baseline diagonal jump. None identifies the assistant as human; the
human-source interpretation requires the explicit premises in
Appendix~\ref{sec:assistance}.

\paragraph{Reproduction and coverage.}
The pinned project is in \texttt{anc/formalization/}.
After installing its declared dependencies,
\texttt{python3 reproduce.py --lean} from \texttt{anc/} checks payload hashes,
finite research checks, the proof build, and the complete project axiom audit.
\texttt{paper-coverage.json} locks the manuscript source hash, inventories all
formal statements, and records unchecked obligations individually.
\texttt{COVERAGE.md} summarizes this inventory. The separate
\texttt{--require-complete} gate fails while any statement or foundation
remains open. Proof-build success does not assert full-paper verification,
human hypercomputation, or empirical performance gains.

\section{Nested limiting revision}
\label{app:nested}

The distinction between architectural meta-depth and nested limit depth can be made exact. Set $A^{(0)}=A$ and $A^{(k+1)}=(A^{(k)})'$.

\begin{definition}[Nested limit representation]
For a fixed positive integer $k$, a $k$-fold nested limit representation over $A$ is a total $A$-computable binary function $g(x,t_1,\ldots,t_k)$ such that every iterated inner limit exists for every fixed choice of the remaining arguments, and the outer limit also exists. The represented set $D$ satisfies
\[
\chi_D(x)=\lim_{t_1\to\infty}\lim_{t_2\to\infty}\cdots
\lim_{t_k\to\infty}g(x,t_1,\ldots,t_k).
\]
\end{definition}

\begin{proposition}[Iteration of the classical limit lemma]
\label{prop:nested}
For every fixed $k\geq1$, a set has a $k$-fold nested limit representation over $A$ if and only if it is computable from $A^{(k)}$.
\end{proposition}

\begin{proof}
For $k=1$ this is Theorem~\ref{thm:limit}. Assume the statement for $k$, including binary functions of several variables via a computable pairing of their arguments.

For the forward direction at $k+1$, take the $k$ inner limits and regard $(x,t_1)$ as the paired input. By the induction hypothesis their value is a total $A^{(k)}$-computable binary function $h(x,t_1)$. The outer limit is then computable from $(A^{(k)})'=A^{(k+1)}$ by the relativized limit lemma.

Conversely, if $D\T A^{(k+1)}$, that lemma supplies a total $A^{(k)}$-computable binary approximation $h(x,t_1)$ converging pointwise to $\chi_D(x)$. Apply the induction hypothesis to the characteristic set of paired inputs on which $h$ equals $1$. It gives one total $A$-computable $k$-fold nested representation, uniformly in $(x,t_1)$, with every required inner limit existing. Adding the outer limit yields the stated $(k+1)$-fold representation.
\end{proof}

The uniformity across all remaining arguments is part of the definition. Independently choosing a different approximation for each input is insufficient. The result characterizes representability, not strict degree growth for every represented set: an $A$-computable set admits a constant representation at every depth. It also requires actual nested limits; a finite run of a recursively generated architecture continues to satisfy Theorem~\ref{thm:simulation}.

\section{Experimental protocol and reproducibility}
\label{app:experiment-protocol}

\paragraph{Design and randomization.}
The configuration, methods, and evaluation criteria were fixed before the reported run. The seed is 2601009, with disjoint deterministic substreams for task generation, stopping-rule validation, and policy mutation. Subspaces are generated by drawing uniform nonzero vectors and retaining independent ones. The stopping experiment uses 1,800 independent tasks and six methods, yielding 10,800 method runs. The policy experiment uses 80 independent replicates of six episodes and six methods, yielding 2,880 executions. Dimension-specific results are retained alongside aggregates. The runner uses Python 3.10.9 and its standard library; analysis uses NumPy 2.2.6 and Matplotlib 3.10.9. Execution of both experiments took 5.52 seconds on an ARM64 macOS host, excluding analysis and plotting. These small synthetic workloads do not estimate the cost of deployed AI agents.

\paragraph{Solver and evaluation.}
Gaussian elimination selects a coefficient consistent with acquired labels, setting free variables to zero. All i.i.d. acquisition queries, including redundant ones, are charged. Evaluation of risk and coverage uses the full generative distribution. If $e=\theta\oplus\widehat\theta$, the exact risk is
\[
R(\widehat\theta)=\frac{1-\rho}{2}
 \mathbf 1\{\exists h\in H:e^\top h=1\}
 +\frac{\rho}{2}\mathbf 1\{e\ne0\}.
\]
For a basis of $H$, testing those basis vectors suffices for the first indicator. The mixture components include zero. The formulas are checked against enumeration of every vector in dimensions one through eight. Separately, the original finite verification suite checks 90 feature spans, 62,566 answer-identification cases, and 40 guarded proposer runs. These finite checks do not establish the uniform undecidability results.

\paragraph{Stopping details and confidence.}
Prediction stability starts with the zero solver and compares successive predictions on the same independent unlabeled pool. Two consecutive equal comparisons are required, with a maximum of eight looks. The statistical rule uses a fresh batch at each look and never incorporates validation labels into fitting or training-span statistics. For a frozen solver with $R>\epsilon$, the conditional probability of zero errors in $m$ independent validation trials is
\[
(1-R)^m\leq(1-\epsilon)^m\leq\delta/8,
\qquad
m=\left\lceil\frac{\log(\delta/8)}{\log(1-\epsilon)}\right\rceil=252.
\]
A union bound over the eight possible looks gives error probability at most $\delta=.05$ even though the solver and stopping time depend on previous observations. It gives no all-input guarantee. At an unsuccessful cap, the teacher is retained; such trials contribute to cost and removal-rate denominators, while risk conditional on removal excludes them. A validation transcript could supply additional identification information; this simple stopping-only baseline does not exploit it. Neither the statistical nor the prediction-stability implementation is presented as a reproduction of a published stopping algorithm.

\paragraph{Policy search and cost.}
Greedy groups the current unlabeled inputs by their residual modulo the query span and chooses a largest nonzero group, breaking ties by vector value. If no such group exists, it uses a supplied basis vector. The program-mutation search uses 64 proposals plus one initial program evaluation per episode. Swap and replacement proposals each have probability $1/2$; replacements are sampled from the current unlabeled pool. Each candidate is interpreted with the same rank guard, and strict empirical improvements are retained. The two mutation methods use identical random edit draws, applied to their respective inherited or reset programs. Cost counters include rejected candidates, protected execution, and sample inspection; development search is reported separately. Elimination XORs count vector operations, not all CPU instructions. In particular, the coordinate basis needs no elimination XORs but still incurs execution and teacher costs.

\paragraph{Uncertainty and artifacts.}
Policy intervals use 2,000 percentile bootstrap resamples of whole independent replicates, preserving within-replicate episodes. Paired differences use the same replicate pairing. Curve bands are pointwise, and the reported intervals are descriptive without a multiple-comparison adjustment. Complete per-task measurements, dimension-specific summaries, executed programs, accepted mutations, configuration, and code/data hash manifests accompany the experimental source. The experiments implement persistent program revision with a fixed improver. They do not evaluate noisy teachers, changes to the feature representation, or unrestricted natural-language agent tasks.

\begingroup
\small
\setlength{\bibsep}{2pt plus .5pt minus .5pt}
\bibliographystyle{plainnat}
\bibliography{references}
\endgroup
\end{document}